\documentclass[11pt]{article}
\usepackage[bottom=1in,left=1in, right=1in, top=1in]{geometry}
\usepackage{authblk}
\usepackage{algorithm}
\usepackage[noend]{algpseudocode}
\usepackage{tabularx}
\usepackage{enumitem}
\usepackage{graphicx}
\usepackage{subcaption}
\usepackage{amsmath}
\usepackage{amssymb}
\usepackage{amsfonts}
\usepackage{cite}
\usepackage{hyperref}
\usepackage{booktabs}
\usepackage[T1]{fontenc}

\hypersetup{
	colorlinks=true, %set true if you want colored links
	linktoc=all,     %set to all if you want both sections and subsections linked
	linkcolor=blue,  %choose some color if you want links to stand out
	citecolor = blue
}

\usepackage{titlesec}
\titlespacing{\paragraph}{0pt}{5pt}{5pt}

\usepackage{amsthm}

\theoremstyle{plain}
\newtheorem{theorem} {Theorem} [section]
\newtheorem{lemma} [theorem] {Lemma}
\newtheorem{observation} {Observation}

\theoremstyle{definition}

\DeclareMathOperator*{\minimize}{minimize}

\DeclareMathOperator{\argmin}{argmin}

\usepackage{mathtools}

\DeclarePairedDelimiter\floor{\lfloor}{\rfloor}
\newcommand{\mfc}{\textsc{Metric Forest Completion}}

\title{Approximate Tree Completion and Learning-Augmented Algorithms for Metric Minimum Spanning Trees\thanks{This work was performed under the auspices of the U.S. Department of Energy by Lawrence Livermore National Laboratory under Contract DE-AC52-07NA27344 (LLNL-ABS-XXXXXX), and was supported by LLNL LDRD project 24-ERD-024.}}

\author[1]{Nate Veldt}
\author[1]{Thomas Stanley}
\author[2]{Benjamin W. Priest}
\author[2]{Trevor Steil}
\author[2]{Keita Iwabuchi}
\author[2]{T.S.~Jayram}
\author[2]{Geoffrey Sanders}

\affil[1]{Department of Computer Science and Engineering, Texas A\&M University}
\affil[2]{Center for Applied Scientific Computing, Lawrence Livermore National Laboratory}
\date{}

\begin{document}
	
	\maketitle
\begin{abstract}
Finding a minimum spanning tree (MST) for $n$ points 
in an arbitrary metric space is a fundamental primitive for hierarchical clustering and many other ML tasks, but this takes $\Omega(n^2)$ time to even approximate.
We introduce a framework for metric MSTs that first (1) finds a forest of disconnected components using practical heuristics, and then (2) finds a small weight set of edges to connect disjoint components of the forest into a spanning tree. 
We prove that optimally solving the second step still takes $\Omega(n^2)$ time, but we provide a subquadratic 2.62-approximation algorithm. In the spirit of learning-augmented algorithms, we then show that if the forest found in step (1) overlaps with an optimal MST, we can approximate the original MST problem in subquadratic time, where the approximation factor depends on a measure of overlap. In practice, we find nearly optimal spanning trees for a wide range of metrics, while being orders of magnitude faster than exact algorithms.
\end{abstract}

\section{Introduction}
\label{sec:intro}
Finding a minimum spanning tree of a graph is a classical combinatorial problem with well-known algorithms dating back to the early and mid 1900s~\cite{boruvka1926jistem,prim1957shortest,kruskal1956shortest}. A widely-studied special case in theory and practice is the \emph{metric} MST problem, where each node corresponds to a point in a metric space and every pair of points defines an edge with weight equal to the distance between points. Finding a metric MST has widespread applications including network design~\cite{loberman1957formal}, approximation algorithms for traveling salesman problems~\cite{held1970traveling}, and feature selection~\cite{labbe2023dendrograms}. The problem also has a very well-known connection to hierarchical clustering~\cite{gower1969minimum} and has been used as a key step in clustering astronomical data~\cite{barrow1985minimal,march2010fast}, analyzing gene expression data~\cite{xu2002clustering}, document clustering~\cite{xu19972d}, and various image segmentation and classification tasks~\cite{xu19972d,an2000fast,la2022ocmst}.

This paper is motivated by challenges in efficiently computing metric MSTs in modern machine learning (ML) and data mining applications. The most fundamental challenge is simply the massive size of modern datasets. In theory one can always find an MST for $n$ points by computing all $O(n^2)$ distances and running an existing MST algorithm. However, this quadratic complexity is far too expensive both in terms of runtime and memory for massive datasets. 
A second challenge is handling complicated metric spaces. Most existing algorithms for metric MSTs are designed for Euclidean distances or other simple metric spaces, often with a particular focus on small-length feature vectors~\cite{march2010fast,agarwal1990euclidean,shamos1975closest,vaidya1988minimum,arya2016fast,wang2021fast}.
While useful in certain settings, these are limited in their applicability to high-dimensional feature spaces and complex distance functions. Many modern ML tasks focus on non-Cartesian data (e.g., videos, images, text, nodes in a graph, or even entire graphs) which must be classified, clustered, or otherwise compared, possibly after being embedded into some metric space. In these settings, even querying the distance score between two data points becomes non-trivial and often involves some level of uncertainty. 
Computing distances for popular non-Euclidean metrics like Levenshtein distance or graph kernels is far more expensive than computing Euclidean distances.
This motivates a body of research on minimizing the number of queries needed to find a minimum spanning tree~\cite{bateni2024metric,erlebach2022learning,hoffman2008computing,megow2017randomization}. 
%
% Another challenge is simply the massive size of modern datasets, which poses a significant challenge even when querying distances is not a bottleneck. In theory one can always find an MST for $n$ points by computing all $O(n^2)$ distances and running an existing MST algorithm. However, this quadratic complexity is far too expensive both in terms of runtime and memory for massive datasets. 

These applications and challenges motivate new approaches for efficiently finding good spanning trees for a set of points in an arbitrary metric space. Ideally, we would like an algorithm whose memory, runtime, and distance query complexity are all $o(n^2)$, while still being able to find spanning trees with strong theoretical guarantees in arbitrary metric spaces. 
Unfortunately,
there are known hardness results that pose challenges for obtaining meaningful subquadratic algorithms. In particular, it is known that finding any constant factor approximation for an MST in an arbitrary metric space requires 
knowing $\Omega(n^2)$ edges in the underlying metric graph~\cite{indyk1999sublinear}. Overcoming this inherent challenge requires exploring additional assumptions and alternative types of theoretical approximation guarantees.

\paragraph{The present work: metric forest completion.} 
To achieve our design goals and overcome existing hardness results, we take our inspiration from the nascent field of learning-augmented algorithms, also known as algorithms with predictions~\cite{mitzenmacher2022algorithms}. The learning-augmented paradigm assumes access to an ML heuristic that provides a prediction or ``warm start'' that is useful in practice but does not come with a priori theoretical guarantees. The goal is to design an algorithm that (1) comes with improved theoretical guarantees when the ML heuristic performs well, and (2) recovers similar worst-case guarantees if the ML heuristic performs poorly. Improved theoretical guarantees can take various forms, including faster runtimes (e.g., for sorting~\cite{bai2023sorting}, binary search~\cite{lin2022learning} or maximum $s$-$t$ flows~\cite{davies2023predictive}), improved approximation factors (e.g., for NP-hard clustering problems~\cite{ergun2022learning,nguyen2023improved}), better competitive ratios for online algorithms (e.g., for ski rental problems~\cite{shin2023improved}), or some combination of the above (e.g., better trade-offs for space requirements vs.\ false positive rates for Bloom filters~\cite{kraska2018case}). These improved guarantees are given in terms of some parameter measuring the quality of the prediction, which is typically not known in practice but provides a concrete measure of error that can be used in theoretical analysis.

In this paper we specifically assume access to fragments of a spanning tree for a metric MST problem (called the \emph{initial forest}), that takes the form of a spanning tree for each component of some partitioning of the data objects. This input can be interpreted as a heuristic approximation for the forest that would be obtained by running a few iterations of a classical MST algorithm such as Kruskal's~\cite{kruskal1956shortest} or Boruvka's~\cite{boruvka1926jistem}. Given this input, we formalize the \textsc{Metric Forest Completion} problem (MFC), whose goal is to find a minimum-weight set of edges that connects disjoint components to produce a full spanning tree.
% An initial forest can be efficiently obtained in a number of different ways, including fast clustering algorithms or by finding connected components in an approximate $k$-nearest neighbors graph. 
Although optimally solving MFC takes $\Omega(n^2)$ distance queries, we design a subquadratic approximation algorithm for MFC and prove a learning-augmented style approximation guarantee for the original metric MST problem. To summarize, we have the following contributions. 
\begin{itemize}[leftmargin=10pt,itemsep=0pt]
	\item \textbf{New algorithmic framework.} We introduce \textsc{Metric Forest Completion} for large-scale metric MST problems, along with strategies for computing an initial forest and a discussion of how our problem fits into the recent framework of learning-augmented algorithms.
	\item \textbf{Approximate completion algorithm.}  We prove that optimally solving MFC requires $\Omega(n^2)$ edge queries, but provide an approximation algorithm with approximation factor $\approx 2.62$. Our algorithm has a query complexity of $o(n^2)$ as long as the initial forest has $o(n)$ components. 
	\item \textbf{Learning-augmented approximation guarantees.} We prove that if the initial forest $\gamma$-overlaps with an optimal MST, then our algorithm is a $(2\gamma+1)$-approximation algorithm for the metric MST problem. If $\gamma = 1$, this means all edges in the forest are contained in an optimal MST. A precise definition for $\gamma > 1$ is given in Section~\ref{sec:mfc}.
	\item \textbf{Experiments.} We show that our method is extremely scalable and obtains very good results on synthetic and real datasets. We also show that simple heuristics provide initial forests with $\gamma$-overlap values that are typically smaller than 2 on various datasets and distance functions, including many non-Euclidean metrics.
\end{itemize}
% We remark finally that although our approximation proofs are somewhat technical, the algorithm itself is simple to implement and is highly parallelizable. 
% This highlights promising opportunities for future research on fast parallelized methods for finding spanning trees in arbitrary metric spaces.
	
\section{Technical Preliminaries}
\label{sec:prelims}
We cover several technical preliminaries to set the stage for our new MFC problem and algorithms. 

\subsection{Graph notation and minimum spanning trees.}
For an undirected graph $G = (V,E)$ and weight function $w \colon E \rightarrow \mathbb{R}^+$, we denote the weight of an edge $e = (u,v)$ as $w(u,v)$, $w_{uv}$, or $w_e$. 
% Every weight function $w$ can be extended to apply to entire sets of edges, so that 
The weight of an edge set $F \subseteq E$ is denoted by 
 \begin{equation}
 	\label{eq:edgesum}
 	w(F) = \sum_{e \in F} w_e.
 \end{equation}
% The weight of all edges in $G$ is denoted by 
We often use $w(G) = w(E)$
to denote the weight of all edges in $G = (V,E)$. 

\paragraph{Minimum spanning trees.} 
The minimum spanning tree problem on $G$ with respect to weight function $w$ seeks a spanning tree $T = (V,E_T)$ of $G$ that minimizes $w(E_T)$. When $w$ is clear from context we will refer to $T$ simply as an MST of $G$. We do often consider multiple spanning trees of the same graph, each optimal for a different weight function. In these cases we  explicitly state the weight function associated with an MST.

There are many well-known greedy algorithms for optimally constructing an MST~\cite{kruskal1956shortest, prim1957shortest,boruvka1926jistem}. 
We review Kruskal's~\cite{kruskal1956shortest} and Boruvka's~\cite{boruvka1926jistem} algorithms as their mechanics are relevant for understanding our approach and results. These methods first place all nodes into singleton components. Each iteration of Kruskal's algorithm identifies a minimum weight edge that connects two nodes in different components, and adds it to a forest of edges (initialized to the empty set at the outset of the algorithm) that is guaranteed to be part of some MST. This proceeds until all components of the forest have been merged into one tree. The method is equivalent to ordering all edges based on weight and visiting them in order, greedily adding the $i$th edge to the spanning tree if and only if its endpoints are in different components. Boruvka's algorithm~\cite{boruvka1926jistem} is similar, but instead identifies a minimum weight edge incident to  \emph{each} component, adding all of them to the growing forest. Both algorithms can be implemented to run in $O(m \log n)$ time where $n = |V|$ and $m = |E|$. Faster algorithms, often based on one of these algorithms, have been developed. The fastest deterministic algorithm has a runtime of $O(m \cdot \alpha(m,n))$ where $\alpha$ is the inverse of the Ackerman function~\cite{chazelle2000minimum}. For this paper it suffices to know that finding an MST takes $\tilde{O}(m)$ time where $\tilde{O}$ hides factors that are logarithmic (or smaller) in $n$.

\subsection{The metric MST problem}
Throughout the paper we let $(\mathcal{X},d)$ be a finite metric space where $\mathcal{X} = \{x_1, x_2, \hdots , x_n\}$ is a set of data points and $d(x_i, x_j)$ is the distance between the $i$th and $j$th points.  
In order to be a metric space, $d$ must satisfy the triangle inequality: $d(x_{i}, x_{j}) \leq d(x_{i}, x_{k}) + d(x_{k}, x_{j})$ for all triplets $i,j,k \in [n] = \{1,2, \hdots n\}$. Aside from this we make no formal assumptions about the points in $\mathcal{X}$ or the distance function $d$.
The space is associated with a complete graph $G_\mathcal{X} = (\mathcal{X}, E_\mathcal{X})$ where we treat $\mathcal{X}$ as a node set and the edge set $E_\mathcal{X} = {\mathcal{X} \choose 2}$ includes all pairs of nodes. 
The weight function $w_\mathcal{X}$ for $G_\mathcal{X}$ is defined by $w_\mathcal{X}(i,j) = d(x_i, x_j)$. This graph $G_\mathcal{X}$ is implicit; in order to know the weight of an edge we must query the distance function $d$. We make no assumptions regarding the complexity of querying distances. Rather, we will give the complexity of our algorithms both in terms of the runtime and number of queries. 

The \emph{metric} minimum spanning tree problem is simply the MST problem applied to $G_\mathcal{X}$ (see Figure~\ref{fig:metric_mst}). This can be solved by querying the distance function $O(n^2)$ times to form $G_\mathcal{X}$, and then applying an existing MST algorithm. 
We show it is more practical to implicitly deal with $G_\mathcal{X}$ via edge queries, while solving different types of nearest neighbor problems over subsets of $\mathcal{X}$.
\begin{figure}[t]
	\centering
	\begin{subfigure}[b]{0.32\textwidth}
		\centering
		\includegraphics[width=\textwidth]{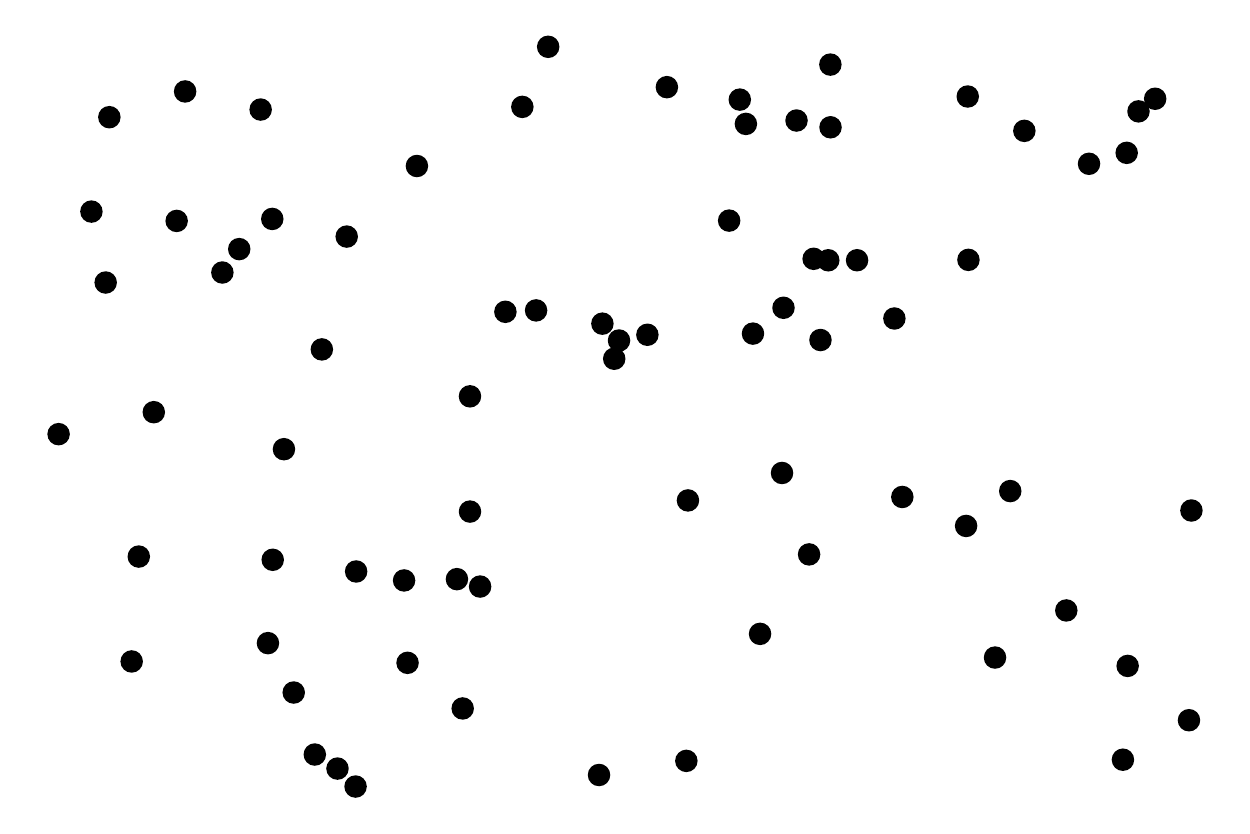}
		\caption{Initial space $\mathcal{X}$}
		\label{fig:opt_mstc}
	\end{subfigure}
	\begin{subfigure}[b]{0.33\textwidth}
		\centering
		\includegraphics[width=\textwidth]{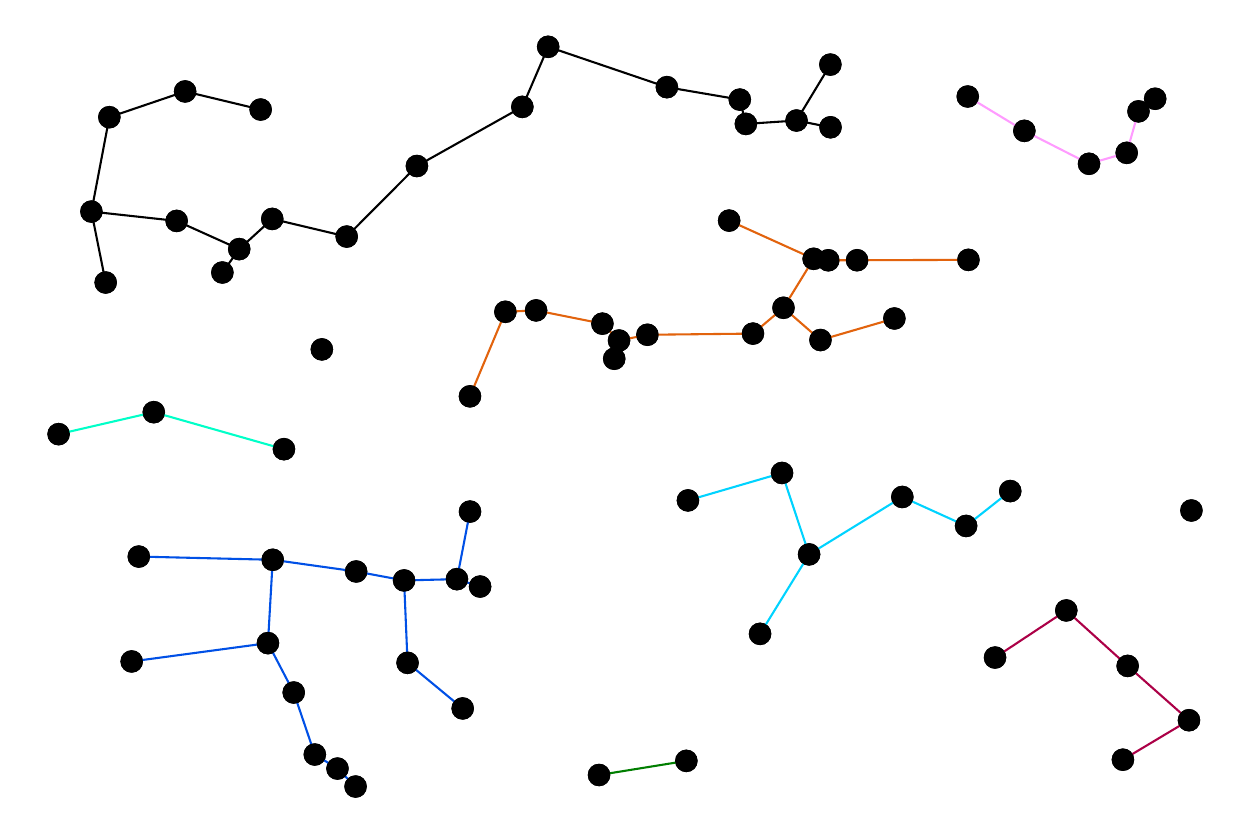}
		\caption{Forest from Kruskal's}
		\label{fig:partial_MST_1}
	\end{subfigure}
	\begin{subfigure}[b]{0.33\textwidth}
		\centering
		\includegraphics[width=\textwidth]{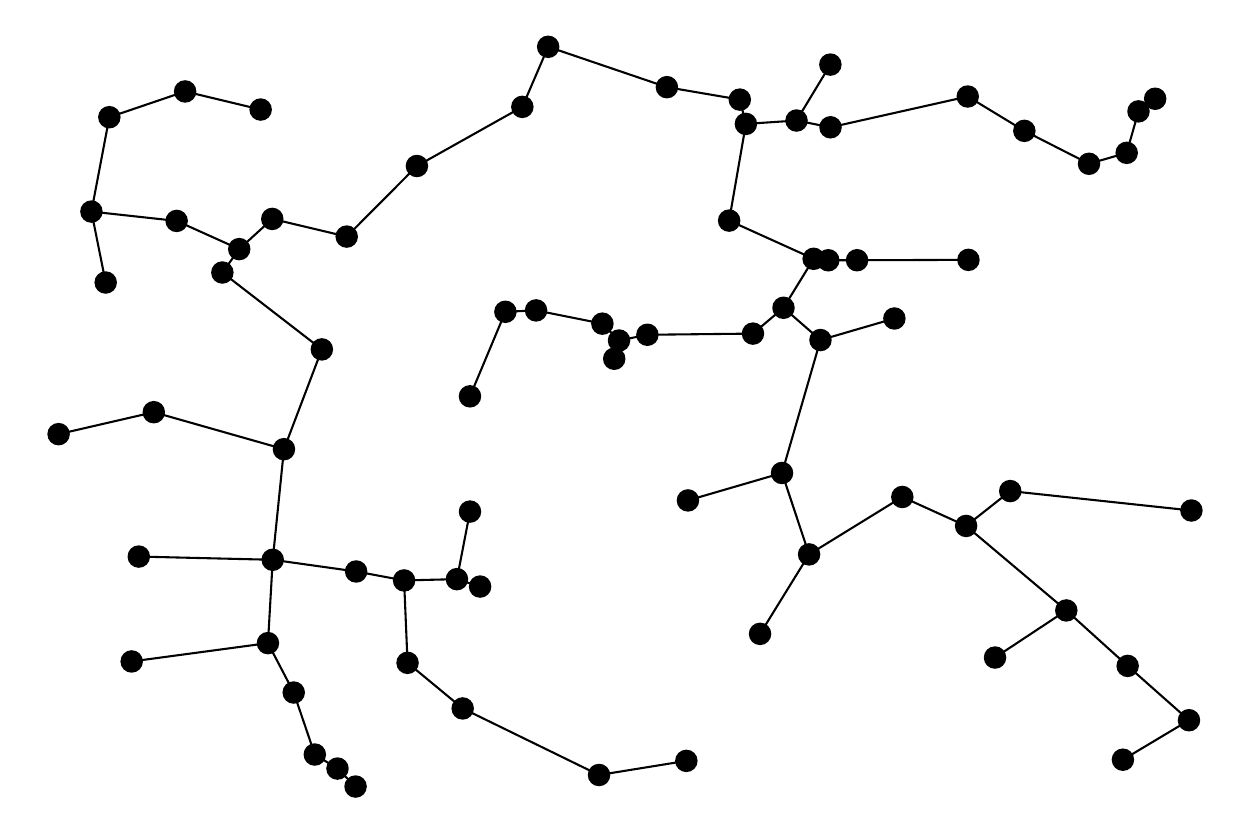}
		\caption{Full metric MST}
		\label{fig:true_MST_1}
	\end{subfigure}
	\caption{(a)~Consider a simple example of a finite metric space $(\mathcal{X},d)$: 75 points in $\mathbb{R}^2$ equipped with Euclidean distance. $G_\mathcal{X}$ is an implicit complete graph obtained by computing distances between all pairs of points. (b)~Kruskal's algorithm iteratively merges components in a growing forest. Here we display the forest at an intermediate step. Knowing the minimum distance between two different components requires solving a bichromatic closest pair problem. (c)~Continuing until all components are merged produces a metric minimum spanning tree.}
	\label{fig:metric_mst}
\end{figure}

\paragraph{Metric MSTs and bichromatic closest pairs.} 
The \emph{bichromatic closest pair problem} (BCP) is a particularly relevant computational primitive for metric MSTs. The input to BCP is two sets of points $A$ and $B$ in a metric space. The goal is to find a pair of opposite-set points (one from $A$ and one from $B$) with smallest distance. One can implicitly apply a classical MST algorithm such as Kruskal's or Boruvka's algorithm to $G_\mathcal{X}$ by repeatedly solving BCP problems. In each step, these algorithms must identify one or more disconnected components in an intermediate forest (see Figure~\ref{fig:partial_MST_1}) to join via minimum weight edges. Finding a minimum weight edge between two components exactly corresponds to a BCP problem. This connection between metric MSTs and BCP is well known and has been leveraged in many prior works on metric MSTs~\cite{agarwal1990euclidean,callahan1993faster,narasimhan2001geometric,chatterjee2010geometric}.

\section{Metric Forest Completion}
\label{sec:mfc}
We now formalize our \mfc{} (MFC) framework\footnote{This is distinct from two other MST-related concepts that also use the acronym MFC; see Section~\ref{sec:related} for details.}
which assumes access to an initial forest that is then grown into a full spanning tree.

 \subsection{Formalizing the MFC problem}
\label{sec:definemfc}
As a starting point for the metric MST problem on $(\mathcal{X},d)$, we assume access to a partitioning $\mathcal{P} = \{P_1, P_2, \hdots , P_t\}$ where $\mathcal{X} = \bigcup_{i = 1}^t P_i$ and $P_i \cap P_j = \emptyset$ for $i \neq j$. For each component $P_i$ we have a partition spanning tree $T_i = (P_i, E_{T_i})$ for that component.
% , which may or may not be an optimal MST for $P_i$. 
See Figure~\ref{fig:warmstart} for an illustration. Let $G_t = (\mathcal{X}, E_t)$ represent the union of these trees, which has the same node set as $G_\mathcal{X}$, and edge set $E_t = \bigcup_{i = 1}^t E_{T_i}$. Each set $P_i$ for $i \in [t]$ defines a group of points in $\mathcal{X}$ as well as a connected component of $G_t$. We refer to this as the \emph{initial forest} for MFC. To provide intuition, the initial forest can be viewed as a proxy for the forest obtained at an intermediate step of Kruskal's or Boruvka's algorithm (see Figure~\ref{fig:partial_MST_1}). While this serves as a useful analogy, we stress that the partitioning will typically be obtained using much cheaper methods and will not satisfy any formal approximation guarantees. Section~\ref{sec:initialforest} covers practical considerations about strategies, runtimes, and quality measures for an initial forest. For now we simply assume it is given as a ``good enough'' starting point, that will ideally overlap, even if not perfectly, with some true MST (see Figure~\ref{fig:overlap}). 
\mfc{} seeks to connect the initial forest into a spanning tree for $G_\mathcal{X}$. Let $P(x) \in \mathcal{P}$ denote the component that $x \in \mathcal{X}$ belongs to. The set of inter-component edges is
\begin{equation*}
	\label{eq:intercomponent}
	\mathcal{I} = \{ (x, y) \in \mathcal{X} \times \mathcal{X} \colon P(x) \neq P(y) \}.
\end{equation*}
We wish to find a minimum weight set of edges $M \subseteq \mathcal{I}$ so that $M\cup E_t$ defines a connected graph on $\mathcal{X}$. If $M$ satisfies these constraints we say it is a valid \emph{completion set} and that $M$ \emph{completes} $\mathcal{P}$. 
The MFC problem can then be written as
\begin{equation}
	\label{eq:mfc}
	\begin{array}{ll}
		\minimize & w_\mathcal{X}(M) + w_\mathcal{X}(E_t)\\
		\text{subject to} & M \emph{ completes } \mathcal{P}.
	\end{array}
\end{equation}
Let $M^*$ denote an optimal completion set. The graph $T^* = (\mathcal{X}, E_t \cup M^*)$ is then guaranteed to be a tree (see Figure~\ref{fig:opt_mfc}); if not we could remove edges to decrease the weight while still spanning $\mathcal{X}$. If the initial forest $G_t$ is in fact contained in some optimal spanning tree of $G_\mathcal{X}$ (which would be the case if it were obtained by running a few iterations of Kruskal's or Boruvka's algorithm), then solving MFC would produce an MST of $G_\mathcal{X}$. In practice this will typically not be the case, but the problem remains well-defined regardless of any assumptions about the quality of the initial forest.

\begin{figure}[t]
	\centering
	\begin{subfigure}[b]{0.32\textwidth}
		\centering
		\includegraphics[width=\textwidth]{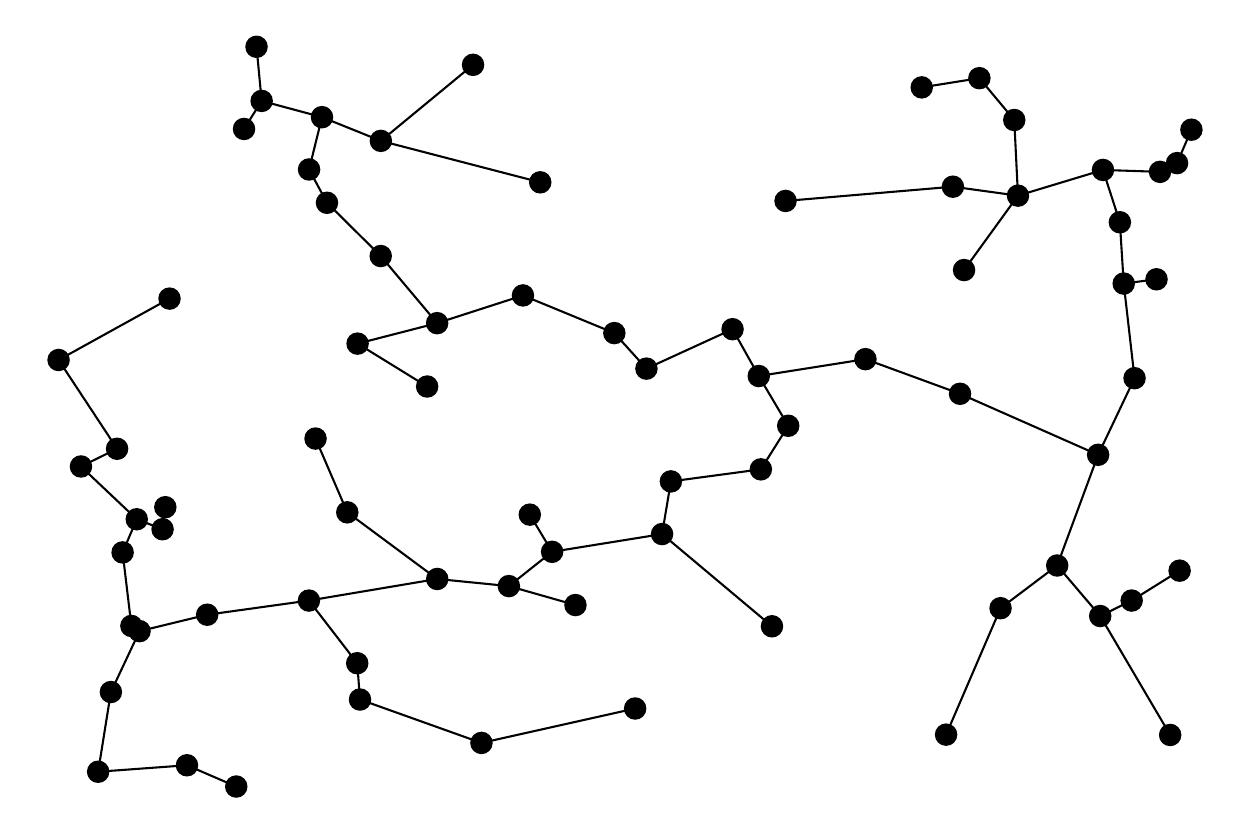}
		\caption{True metric MST}
		\label{fig:truemst}
	\end{subfigure}
	\begin{subfigure}[b]{0.32\textwidth}
		\centering
		\includegraphics[width=\textwidth]{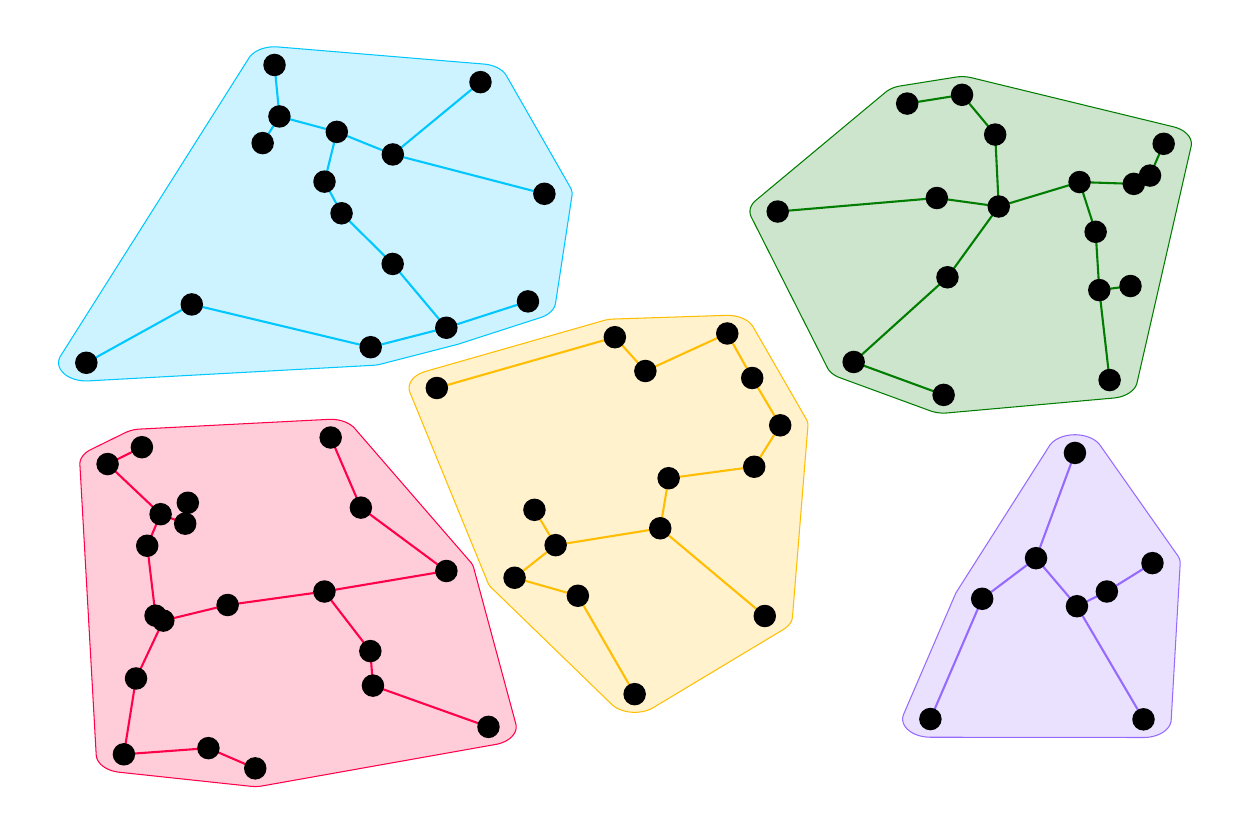}
		\caption{Partial spanning tree ($\mathcal{P}$; $\{T_i\})$}
		\label{fig:warmstart}
	\end{subfigure}
	\begin{subfigure}[b]{0.32\textwidth}
		\centering
		\includegraphics[width=\textwidth]{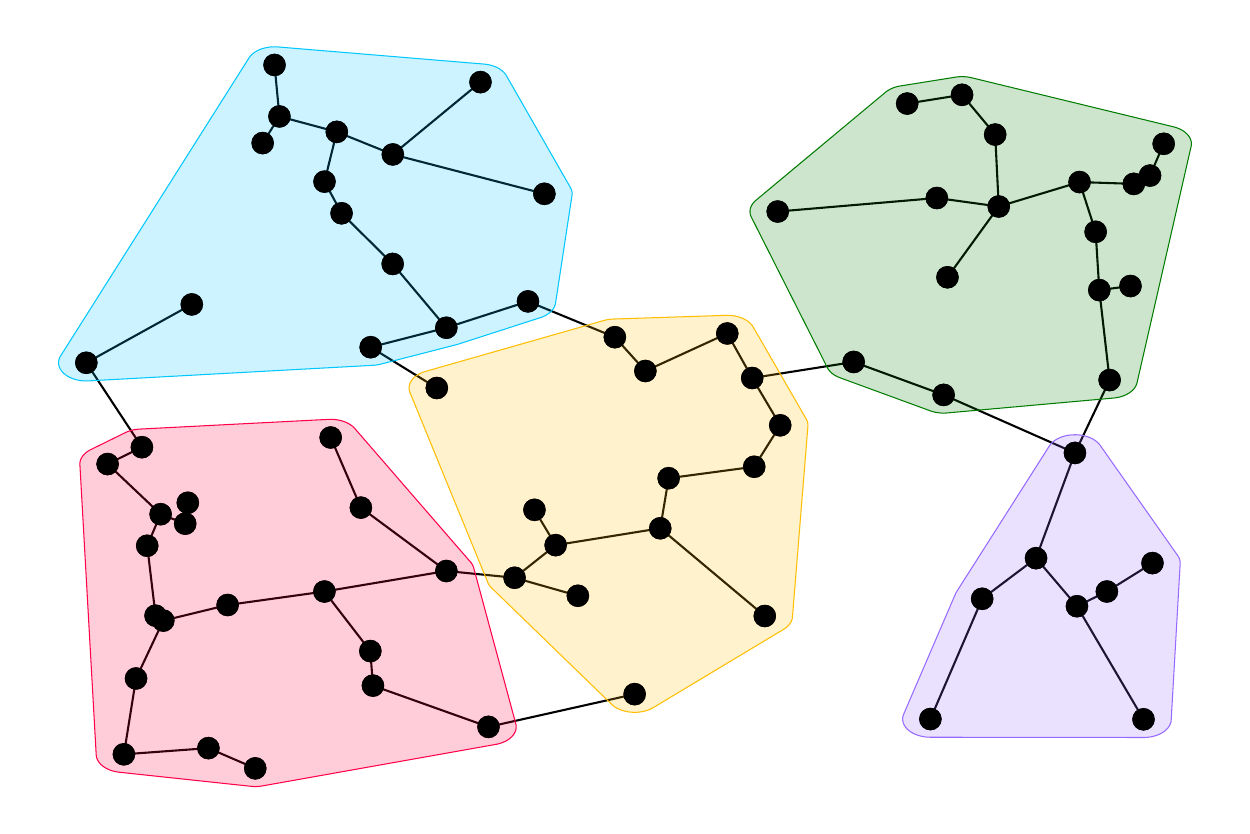}
		\caption{True MST overlap with $\mathcal{P}$}
		\label{fig:overlap}
	\end{subfigure}
	\begin{subfigure}[b]{0.32\textwidth}
		\centering
		\includegraphics[width=\textwidth]{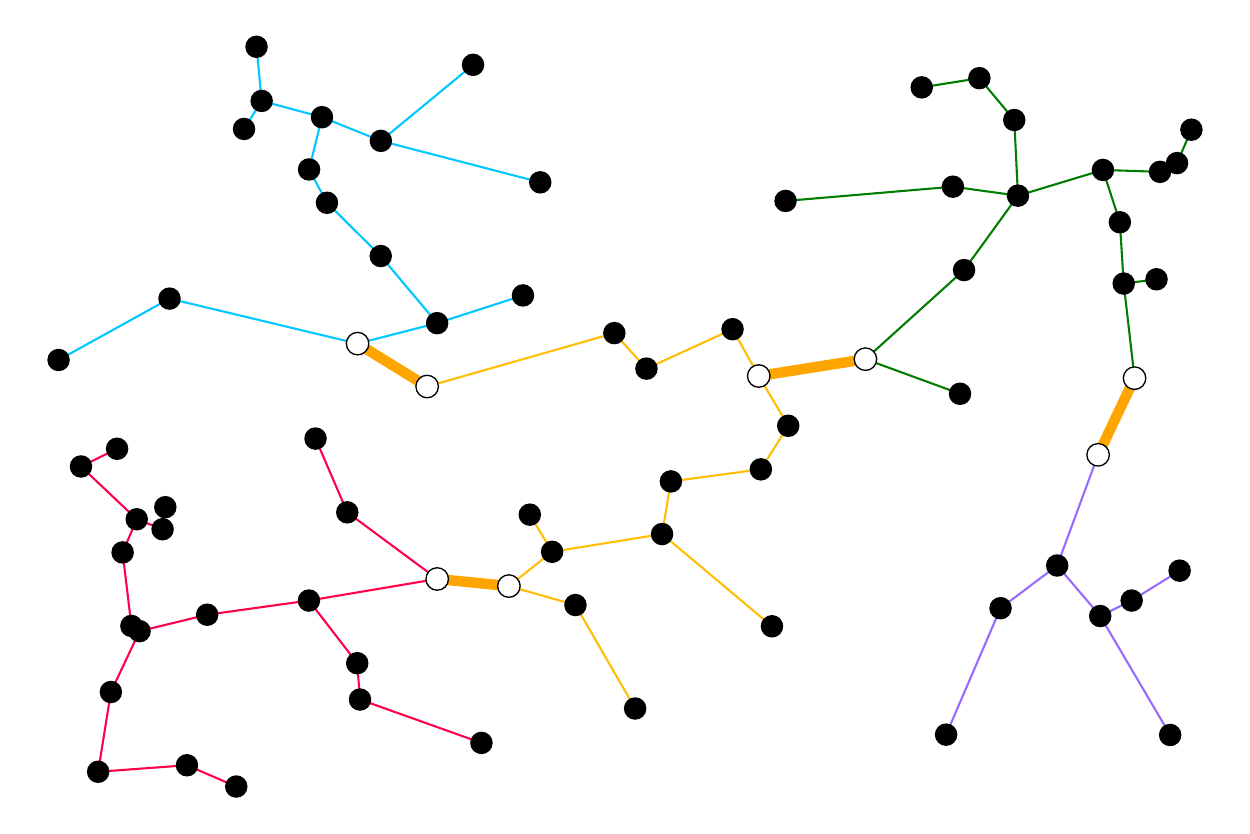}
		\caption{MFC solution ($M^*$ in orange)}
		\label{fig:opt_mfc}
	\end{subfigure}
	\begin{subfigure}[b]{0.32\textwidth}
		\centering
		\includegraphics[width=\textwidth]{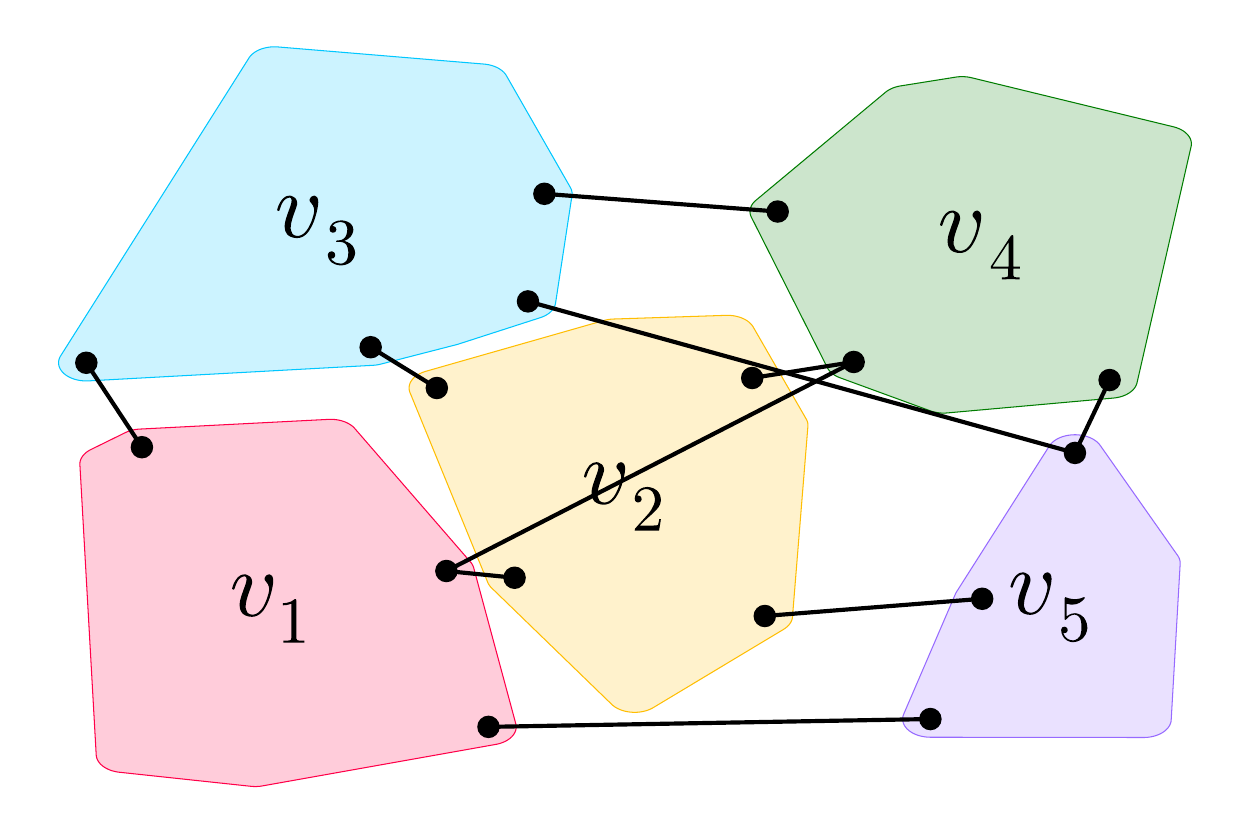}
		\caption{Coarsened graph $G_\mathcal{P}$ w.r.t.\ $w^*$}
		\label{fig:coarsenedgraph}
	\end{subfigure}
	\begin{subfigure}[b]{0.32\textwidth}
		\centering
		\includegraphics[width=\textwidth]{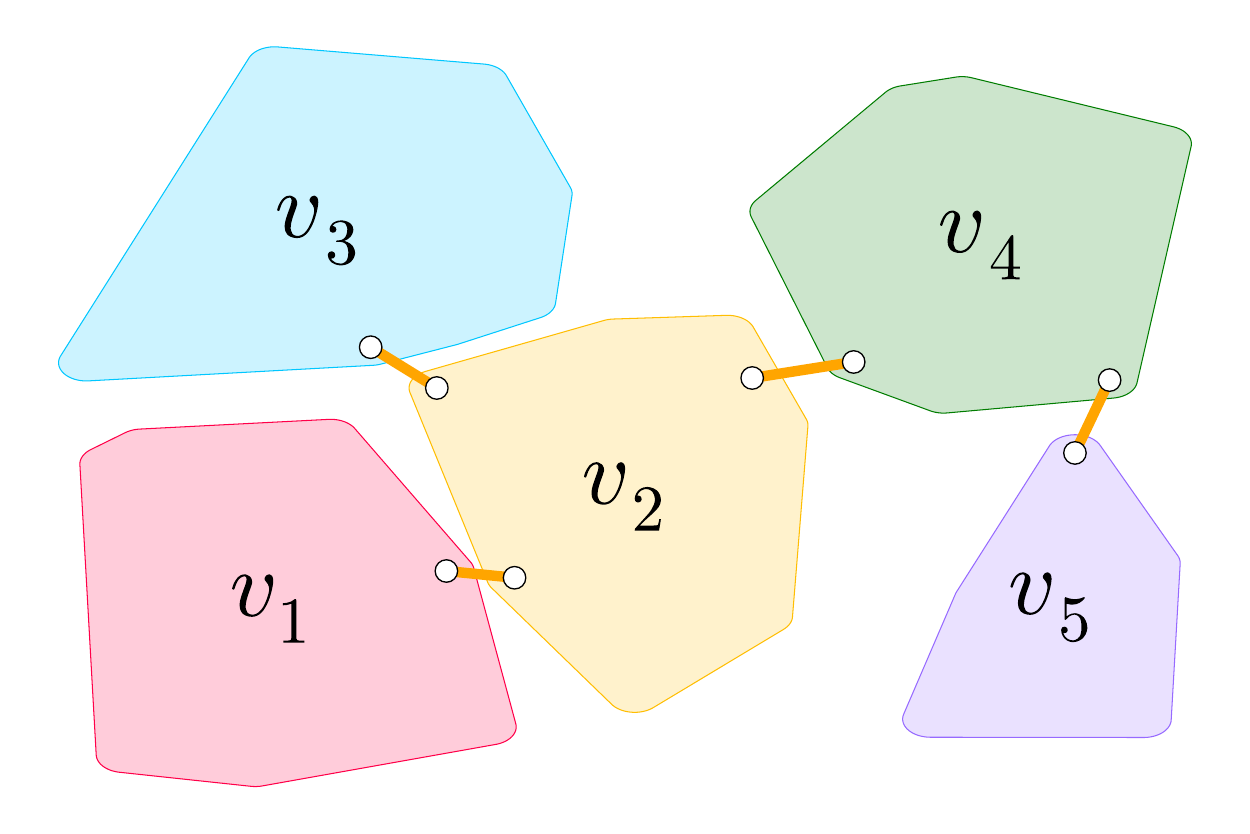}
		\caption{MST of $G_\mathcal{P}$ w.r.t.\ $w^*$}
		\label{fig:mstcoarsened}
	\end{subfigure}
	\caption{(a) We display an optimal metric MST for a toy example with $|\mathcal{X}| = 75$ points. Our framework and algorithm apply to general metric spaces, but for visualization purposes our figures focus on 2-dimensional Euclidean space. (b) The \mfc{} problem is given a partitioning $\mathcal{P}$ and spanning trees $\{T_i\}$ for components of the partition. For this illustration we used a $k$-means algorithm with $k = 5$ computed optimal spanning trees of components using the naive approach. (c)~The true MST overlaps significantly with the initial partial spanning tree, but its induced subgraph on each component is not necessarily connected. For this example, the $\gamma$-overlap (see Section~\ref{sec:learningaugmented}) is $\gamma \leq 1.12$. (d)~The optimal completion set $M^*$ is shown in orange; combining it with the spanning trees of the partial spanning tree produces a spanning tree for all of $\mathcal{X}$. (e)~The coarsened graph $G_\mathcal{P}$ has a node $v_i$ for each component $P_i \in \mathcal{P}$. Solving $O(t^2)$ bichromatic closest pair problems identifies the closest pair of points between each pair of clusters, defining an optimal weight function $w^*$ on $G_\mathcal{P}$. (f)~Finding the minimum-weight completion set $M^*$ amounts to finding the MST of $G_\mathcal{P}$ with respect to weight function $w^*$. }
	\label{fig:three_figures}
\end{figure}

The objective function in Eq.~\eqref{eq:mfc} includes the weight of the initial forest $w_\mathcal{X}(E_t)$. Although this is constant with respect to $M$ and does not affect optimal solutions, there are several reasons to incorporate this term explicitly. Most importantly, our ultimate goal is to obtain a good spanning tree for all of $G_\mathcal{X}$, and thus the weight of the full spanning tree (i.e., the objective in Eq.~\ref{eq:mfc}) is a more natural measure. Considering the weight of the full spanning tree also makes more sense in the context of our learning-augmented algorithm analysis, where the goal is to approximate the original metric MST problem on $G_\mathcal{X}$, under different assumptions about the initial forest. We note finally that excluding the term $w_\mathcal{X}(E_t)$ rules out the possibility of any meaningful approximation results. 
We prove the following result using a reduction from BCP to MFC, combined with a slight variation of a simple lower bound for monochromatic closest pair that was shown in Section 9 of Indyk~\cite{indyk1999sublinear}.
\begin{theorem}
\label{thm:hard}
Every optimal algorithm for MFC has $\Omega(n^2)$ query complexity. Furthermore, for any multiplicative factor $p \geq 1$ (not necessarily a constant), any algorithm that finds a set $M \subseteq \mathcal{I}$ that is feasible for~\eqref{eq:mfc} and satisfies $w_\mathcal{X}(M) \leq p \cdot w_{\mathcal{X}}(M^*)$ requires $\Omega(n^2)$ queries.
\end{theorem}
\begin{proof}
	Let $\mathcal{X}$ be a set of $n$ points that are partitioned into two sets $P_1$ and $P_2$ of size $n/2$. Define a distance function $d$ such that $d(a,b) = 1$ for a randomly chosen pair $(a,b) \in P_1 \times P_2$, and such that $d(x,y) = 2p$ for all other pairs $(x,y) \in {\mathcal{X} \choose 2} \backslash \{(a,b)\}$. Note that this $d$ is a metric. The MFC problem on this instance is identical to solving BCP on $P_1$ and $P_2$. The unique optimal solution is exactly $M^* = (a,b)$, and no other choice of $M \subseteq \mathcal{I}$ comes within a factor $p$ of this solution. Thus, any $p$-approximation algorithm must find the pair $(a,b)$ with distance 1 among a collection of $\Omega(n^2)$ pairs, where all other pairs have distance $2p$. This requires $\Omega(n^2)$ queries.
\end{proof}
Theorem~\ref{thm:hard} shows that it is impossible in general to find optimal solutions for MFC, or multiplicative approximations for $w_\mathcal{X}(M^*)$, in $o(n^2)$ time. However, this does not rule out the possibility of approximating the more relevant objective $w_\mathcal{X}(M) + w_\mathcal{X}(E_t)$.

\paragraph{The MFC coarsened graph.} The  MFC problem is equivalent to finding a minimum spanning tree in a \textit{coarsened graph} $G_\mathcal{P} = (V_\mathcal{P}, E_\mathcal{P})$ with node set $V_\mathcal{P} = \{v_1, v_2, \hdots, v_t\}$ where $v_i$ represents component $P_i$. We refer to $v_i$ as the $i$th \textit{component node}. This graph is complete: $E_\mathcal{P}$ includes all pairs of component nodes. Figure~\ref{fig:coarsenedgraph} provides an illustration of the coarsened graph. Finding an MFC solution $M^* \subseteq \mathcal{I}$ is equivalent to finding an MST in $G_\mathcal{P}$ (see Figure~\ref{fig:mstcoarsened}) with respect to the weight function $w^* \colon E_\mathcal{P} \rightarrow \mathbb{R}^+$ defined for every $i, j \in \{1,2, \hdots, t\}$ by
\begin{equation}
	\label{eq:wstar}
	w^*_{ij} = w^*(v_i, v_j) = d(P_i, P_j) = \min_{x \in P_i, y\in P_j} d(x,y).
	%, \quad \text{ for every $P_i, P_j \in \mathcal{P} \times \mathcal{P}.$}
\end{equation}
Computing $w_{ij}^*$ exactly requires solving a bichromatic closest pair problem over sets $P_i$ and $P_j$. A straightforward approach for computing $w_{ij}^*$ is to check all $|P_i|\cdot |P_j|$ pairs of points in $P_i \times P_j$. The number of distance queries needed to apply this simple strategy to form all of $w^*$ is
\begin{equation*} 
	\frac{1}{2}\sum_{i = 1}^t |P_i| \cdot (n - |P_i|) = \frac{n^2}{2} - \frac{1}{2} \sum_{i = 1}^t |P_i|^2.
\end{equation*}
In a worst-case scenario where component sizes are balanced, we would need
$\Omega\left({t \choose 2} \frac{n}{t} \frac{n}{t}\right) = \Omega(n^2)$ 
queries, which is not surprising given Theorem~\ref{thm:hard}. Nevertheless, this notion of a coarsened graph will be very useful in developing approximation algorithms for MFC.

\subsection{Computing an initial forest}
\label{sec:initialforest}
Our theoretical analysis requires very few assumptions about the partitioning $\mathcal{P}$ and partition spanning trees $\{T_i \colon i = 1,2, \hdots t\}$ that are given as input to the \mfc{} problem. In particular, we need not assume that $T_i$ is a minimum spanning tree or even a good spanning tree of $P_i$. Similarly, the components are not required to be sets of points that are close together in the metric space in order for the problem to be well-defined. Nevertheless, \textit{in practice} the hope is that $\mathcal{P}$ identifies groups of nearby points in $\mathcal{X}$ and that $T_i$ is a reasonably good spanning tree for $P_i$ for each $i \in [t] = \{1,2, \hdots, t\}$. 

In order for our approach to be meaningful for large-scale metric MST problems, we must be able to obtain a reasonably good initial forest without this dominating our overall algorithmic pipeline.
Here we discuss two specific strategies for initial forest computations, both of which are fast, easy to parallelize, and motivated by techniques that are already being used in practice in large-scale clustering pipelines. In particular, there are already a number of existing heuristics for finding MSTs and hierarchical clusters that rely in some way on partitioning an initial dataset and then connecting or merging components~\cite{zhong2015fast,jothi2018fast,mishra2020efficient,chen2013clustering}. See Section~\ref{sec:related} for more details. These typically focus only on point cloud data, do not apply to arbitrary metric spaces, and do not come with any type of approximation guarantee. Nevertheless, they provide examples of fast heuristics for large-scale metric spanning tree problems, and serve as motivation for our more general strategies.

\paragraph{Strategy 1: Components of a $k$-NN graph.}
A natural way to obtain an initial forest for $G_\mathcal{X}$ is to compute an approximate $k$-nearest neighbor graph for a reasonably small $k$, which can be accomplished with the $k$-NN descent algorithm~\cite{dong2011efficient} or a recent distributed generalization of this method~\cite{iwabuchi2023towards}. This efficiently connects a large number of points using small-weight edges. The $k$-NN graph will often be disconnected, and we can use the set of connected components as our initial components $\mathcal{P} = \{P_1, P_2, \hdots , P_t\}$. The exact number of components will depend on the distribution of the data and the number of nearest neighbors computed. For larger values of $k$, the $k$-NN graph is more expensive to compute, but then there are fewer components to connect, so there are trade-offs to consider. The components of the $k$-NN graph will typically not be trees but will be sparse (each node has at most $O(k)$ edges), so we can use classical MST algorithms to find spanning trees for all components in $\sum_{i = 1}^t \tilde{O}(k \cdot |P_i|) = \tilde{O}(kn)$ time. The exact runtime of the $k$-NN descent algorithm depends on various parameters settings, but prior work reports an empirical runtime of $O(n^{1.14})$~\cite{dong2011efficient}, with strong empirical performance across a range of different metrics and dataset sizes. We remark that $k$-NN computations have already been used elsewhere as subroutines for large-scale Euclidean MST computations~\cite{almansoori2024fast,chen2013clustering}.

\paragraph{Strategy 2: Fast clustering heuristics.}
Another approach is to form components of the initial forest $\mathcal{P} = \{P_1, P_2, \hdots, P_t\}$ by applying a fast clustering heuristic to $\mathcal{X}$ such as a distributed $k$-center algorithm~\cite{malkomes2015fast,mcclintock2016efficient}. Even the simple sequential greedy 2-approximation algorithm for $k$-center can produce an approximate clustering using $O(nk)$ queries~\cite{gonzalez1985clustering}. Approximate or exact minimum spanning trees for each $P_i$ can be found in parallel. The remaining step is to find a good way to connect the forest. Similar approaches that partition the initial dataset using $k$-means clustering also exist~\cite{zhong2015fast,jothi2018fast}, though this inherently forms clusters based on Euclidean distances. For all of these clustering-based approaches, the number of components $t$ for the initial forest is easy to control since it exactly corresponds to the number of clusters $k$. Smaller $k$ leads to larger clusters, and hence finding an MST of each $P_i$ is more expensive. However, there are then fewer components to connect, so there is again a trade-off to consider. 
We remark that it may seem counterintuitive to use a clustering method as a subroutine for finding an MST, since one of the main reasons to compute an MST is to perform clustering. We stress that Strategy 2 uses a cheap and fast clustering method that identifies sets of points that are somewhat close in the metric space, without focusing on whether they are good clusters for a downstream application. This speeds up the search for a good spanning tree, which can be used as one step of a more sophisticated hierarchical clustering pipeline. \\

\subsection{MFC as a learning-augmented framework}
\label{sec:learningaugmented}
Our approach fits the framework of learning-augmented algorithms in that the initial forest can be viewed as a prediction for a partial metric MST, such as the forest obtained by running several iterations of a classical MST algorithm. 
In an ideal setting, the initial forest would be a subset of an optimal MST. If so, then an optimal solution to MFC would produce an optimal metric MST. We relax this by introducing a more general way to measure how much an optimal MST ``overlaps'' with initial forest components. 
Let $\mathcal{T}_\mathcal{X}$ represent the set of minimum spanning trees of $G_\mathcal{X}$, and $T \in \mathcal{T}_\mathcal{X}$ denote an arbitrary MST.
For components $\mathcal{P} =  \{P_1, P_2, \hdots, P_t\}$, let $T(P_i)$ denote the induced subgraph of $T$ on $P_i$, and let $T(\mathcal{P}) = \bigcup_{i = 1}^t T(P_i)$ denote the edges of $T$ contained inside components of $\mathcal{P}$. 
% Then, $w(T(\mathcal{P}))$ is the total weight that an optimal MST $T$ places inside components of $\mathcal{P}$. 
Larger values of $w_\mathcal{X}(T(\mathcal{P}))$ indicate better initial forests, since this means an optimal MST places a larger weight of edges inside these components. 
% Given components $\mathcal{P} =  \{P_1, P_2, \hdots, P_t\}$ and corresponding spanning trees $\{T_1, T_2, \hdots, T_t$\}, 
We define the $\gamma$-overlap for the initial forest to be the ratio
\begin{equation}
\label{eq:gamma}
\gamma(\mathcal{P})  = \frac{w_\mathcal{X}(E_t)}{\max_{T \in \mathcal{T}_\mathcal{X}} w_\mathcal{X}(T(\mathcal{P}))}.
\end{equation}
This measures the weight of edges that the initial forest places inside $\mathcal{P}$, relative to the weight of edges an optimal MST places inside $\mathcal{P}$. 
When $\mathcal{P}$ is clear from context we will simply write $\gamma$. 
Lower ratios for $\gamma$ are better. 
In the denominator, we maximize $w_\mathcal{X}(T(\mathcal{P}))$ over all optimal spanning trees since any MST of $G_\mathcal{X}$ is equally good for our purposes; hence we are free to focus on the MST that overlaps most with $\mathcal{P}$. The optimality of $T$ implies that $w_\mathcal{X}(T(P_i)) \leq w_\mathcal{X}(T_i)$ for every $i \in \{1,2, \hdots, t\}$, which in turn implies that $\gamma(\mathcal{P}) \geq 1$ always. Even if $T_i$ is a minimum spanning tree for $P_i$, it is possible to have $w_\mathcal{X}(T(P_i)) < w_\mathcal{X}(T_i)$ since the induced subgraph $T(P_i)$ is not necessarily connected (see Figure~\ref{fig:overlap} for an example). We achieve the lower bound $\gamma = 1$ exactly in the idealized setting where the initial forest is contained in some optimal metric MST. In practice, we expect the clusters $P_i$ and spanning trees $T_i$ to be imperfect in the sense that connecting them will likely not provide an optimal MST for $G_X$. However, for an initial forest where each $P_i$ is a set of nearby points and $T_i$ is a reasonably good spanning tree for $P_i$, we would expect $\gamma$ to be larger than 1 but still not too large. Figure~\ref{fig:overlap} provides an example where we can certify that $\gamma \leq 1.12$ by comparing against one optimal MST. We later show experimentally that we can quickly obtain initial forests with small $\gamma$-overlap for a wide range of datasets and metrics. 

\paragraph{Learning-augmented algorithm guarantees.}
 We typically would not compute the ratio $\gamma$ for large-scale applications as this would be even more computationally expensive than solving the original metric MST problem. However, this serves as a theoretical measure of initial forest quality when proving approximation guarantees, following the standard approach in the analysis of learning-augmented algorithms. 
In the worst case, no algorithm making $o(n^2)$ queries can obtain a constant factor approximation for metric MST in arbitrary metric spaces~\cite{indyk1999sublinear}.  Following the standard goal of learning-augmented algorithms, we will improve on this worst-case setting when the initial forest is good; otherwise we will recover the worst-case behavior. Formally, we will consider the initial forest to be good when $\gamma$ is finite and the number of components is $t = o(n)$. For this setting, Section~\ref{sec:algs} will present a learning-augmented $(2\gamma+1)$-approximation algorithm for metric MST with subquadratic complexity. There are two ways to see that this is no worse than the worst-case guarantee using no initial forest, depending on whether we focus on worst-case runtimes or wort-case approximation factors. From the perspective of runtimes, if $t = \Omega(n)$ we can default to the worst-case quadratic complexity algorithms that provide an optimal MST. From the perspective of approximation factors, we certainly do no worse than the worst-case approximation factor, which is infinite for subquadratic algorithms. We will state our approximation results and runtimes more precisely in Section~\ref{sec:algs}.

\section{Approximate Completion Algorithm}
\label{sec:algs}
\begin{algorithm}[tb]
	\caption{\textsf{MFC-Approx}}
	\label{alg:mfcapprox}
	\begin{algorithmic}[5]
		\State{\bfseries Input:} $\mathcal{X} = \{x_1, x_2, \hdots , x_n\}$, components $\mathcal{P} = \{P_1, P_2, \hdots, P_t\}$, spanning trees $\{T_1, T_2, \hdots, T_t\}$
		\State {\bfseries Output:} Spanning tree for implicit metric graph of $\mathcal{X}$
		\For{$i = 1, 2, \hdots t$}
		\State Select arbitrary component representative $s_i \in P_i$
		\EndFor
		\For{$i = 1, 2, \hdots t-1$}
		\For{$j = i+1, \hdots , t$}
		\State $w_{i \rightarrow j} = \min_{x_i \in P_i} d(x_i, s_j)$  \quad \hfill \texttt{ // closest a $P_i$ node comes to $s_j$}
		\State $w_{j \rightarrow i} = \min_{x_j \in P_j} d(x_j, s_i)$ \quad \hfill  \texttt{ // closest a $P_j$ node comes to $s_i$}
		\State $\hat{w}_{ij} = \min \{w_{i \rightarrow j}, w_{j \rightarrow i}\}$ \quad  \hfill \texttt{// set weight for edge $(v_i, v_j)$} 
		\EndFor
		\EndFor
		\State $\hat{T}_{\mathcal{P}} = \textsc{OptimalMST}(\{\hat{w}_{ij}\}_{i,j \in [t]})$ \hfill \texttt{ // find optimal MST on complete $t$-node graph}
		\State Return spanning tree $\hat{T}$ of $G_\mathcal{X}$ by combining $\bigcup_{i=1}^t T_i$ with edges from $\hat{T}_{\mathcal{P}}$.
	\end{algorithmic}
\end{algorithm}

We now present an algorithm that approximates MFC to within a factor $c < 2.62$. We also prove it can be viewed as a learning-augmented algorithm for metric MST, where the approximation factor depends on the $\gamma$-overlap of the initial forest. 
%Pseudocode for our algorithm is provided in the appendix. Here in the main text we give a full description of the algorithm along with visual aids in Figure~\ref{fig:mfcapprox}. Due to space constraints, proofs are relegated to the appendix.

\begin{figure*}[t]
	\centering
	\begin{subfigure}[b]{0.33\textwidth}
		\centering
		\includegraphics[width=\textwidth]{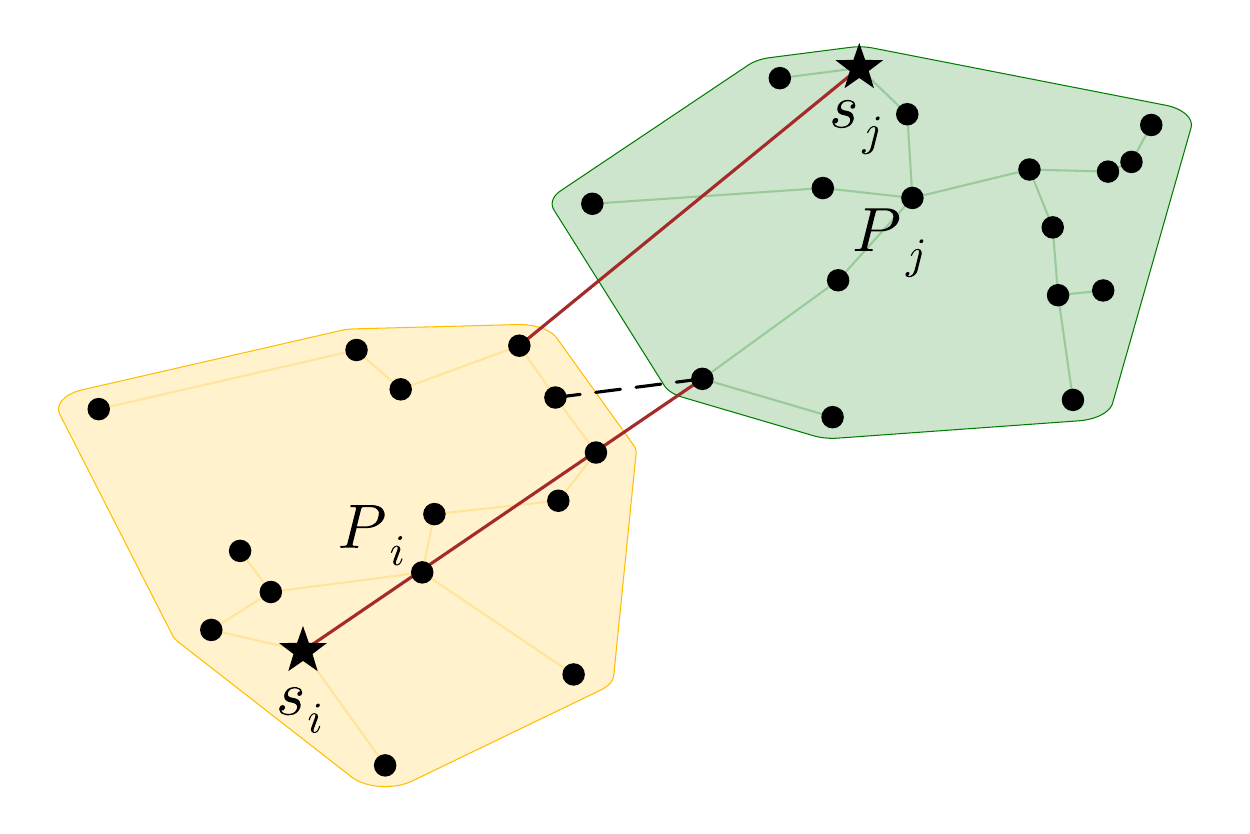}
		%        \vspace{-10pt}
		\caption{Computing $\hat{w}_{ij}$}
		\label{fig:wijhat}
	\end{subfigure}\hfill
	\begin{subfigure}[b]{0.33\textwidth}
		\centering
		\includegraphics[width=\textwidth]{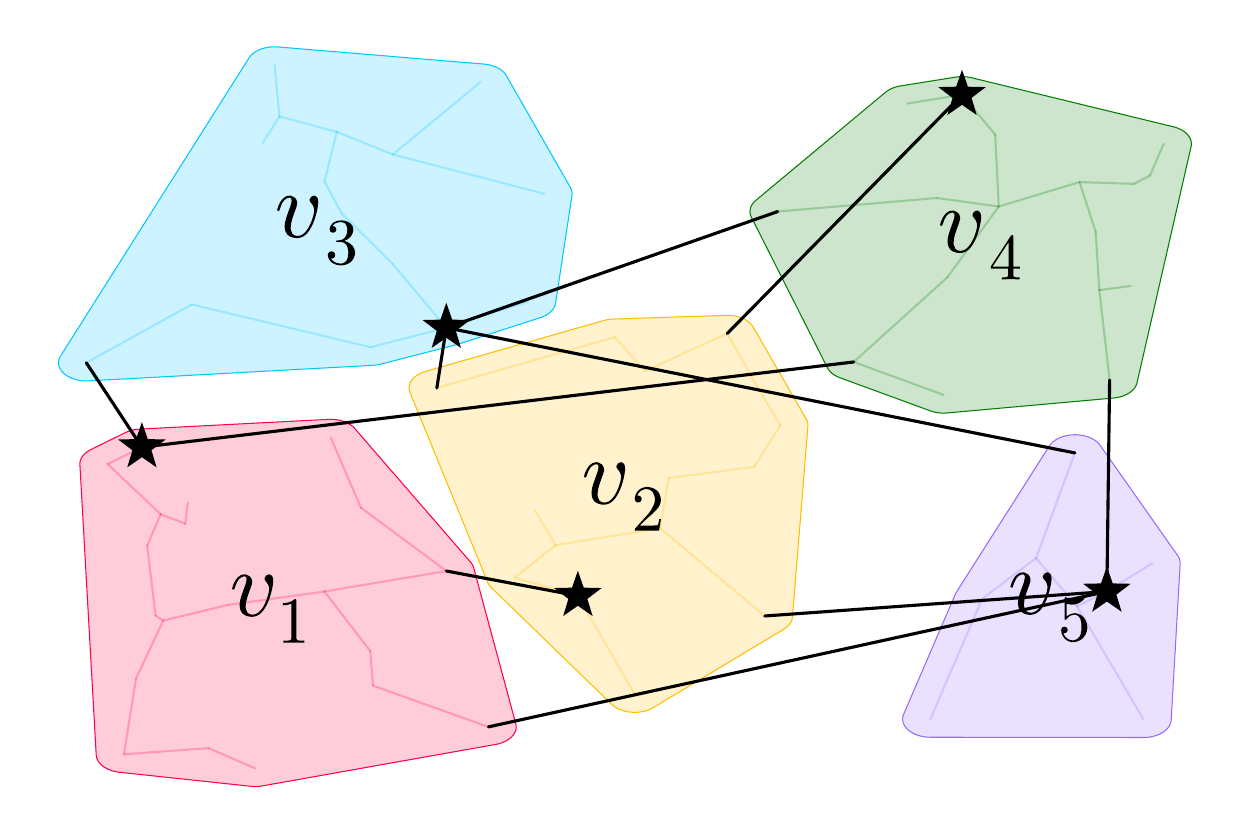}
		%        \vspace{-10pt}
		\caption{$G_\mathcal{P}$ w.r.t.\ $\hat{w}$}
		\label{fig:componentsgraphhat}
	\end{subfigure}\hfill
	\begin{subfigure}[b]{0.33\textwidth}
		\centering
		\includegraphics[width=\textwidth]{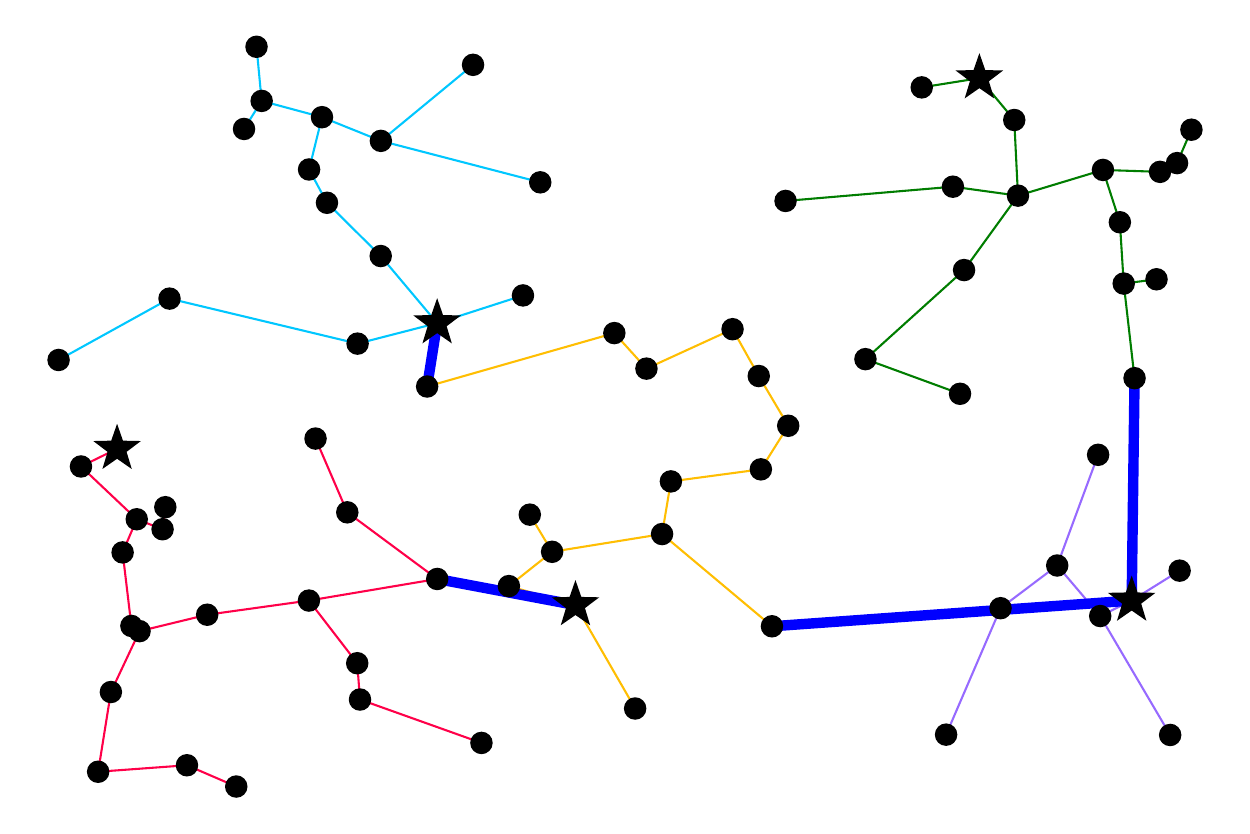}
		%         \vspace{-10pt}
		\caption{\textsf{MFC-Approx} output}
		\label{fig:approx_mfc}
	\end{subfigure}
	\caption{(a) Finding the minimum distance between components $P_i$ and $P_j$ (dashed line) is an expensive bichromatic closest pair problem. \textsf{MFC-Approx} instead performs a cheaper nearest neighbor query for a \emph{representative} point in each component ($s_i$ and $s_j$, shown as stars). The algorithm finds the closest point to each representative from the opposite cluster, then takes the minimum of the two distances. 
		(b) Applying this to each pair of components produces a weight function $\hat{w}$ for the coarsened graph $G_\mathcal{P}$. Finding an MST of $G_\mathcal{P}$ with respect to $\hat{w}$ yields (c) a 2.62-approximation for MFC.}
	\label{fig:mfcapprox}
\end{figure*}

%\subsection{Algorithm description} 
Pseodocode for our method, which we call \textsf{MFC-Approx}, is shown in Algorithm~\ref{alg:mfcapprox}. This algorithm starts by choosing an arbitrary point $s_i \in P_i$ for each $i \in \{1,2, \hdots, t\}$ to be the component's \emph{representative} (starred nodes in Figure~\ref{fig:mfcapprox}). 
The algorithm computes the distance between every point $x \in \mathcal{X}$ and all of the other representatives. For every pair of distinct components $i$ and $j$ we compute the weights:
\begin{align}
	w_{i \rightarrow j} &= \min_{x_i \in P_i} d(x_i, s_j)  	&\text{(the closest $P_i$ node to $s_j$)} \\
	w_{j \rightarrow i} &= \min_{x_j \in P_j} d(x_j, s_i)  &\text{(the closest $P_j$ node to $s_i$).} 
\end{align}
We then define the approximate edge weight between component nodes $v_i$ and $v_j$ to be:
\begin{align}
	\label{eq:approxweight}
	\hat{w}_{ij} = \min \{w_{i \rightarrow j}, w_{j \rightarrow i}\}.
\end{align}
This upper bounds the minimum distance $w_{ij}^*$ between the two components (Figure~\ref{fig:wijhat}).
Computing this for all pairs of components creates a new weight function $\hat{w}$ for the
coarsened graph $G_\mathcal{P}$ (Figure~\ref{fig:componentsgraphhat}). The algorithm keeps track of the points in $\mathcal{X}$ that define these edge weights in $G_\mathcal{P}$. It then computes an MST in $G_\mathcal{P}$ with respect to $\hat{w}$, then identifies the corresponding edges in $G_\mathcal{X}$, to produce a feasible solution for MFC (Figure~\ref{fig:approx_mfc}).

\subsection{Bounded and unbounded edges in the coarsensed graph}
A key step of our approximation analysis is to separate edges of the coarsened graph $G_\mathcal{P} = (V_\mathcal{P},E_\mathcal{P})$ into two categories, depending on the relationship between $\hat{w}$ and $w^*$.
% \begin{definition}
Formally, for an arbitrary constant $\beta \geq 1$, we say edge $(v_i,v_j) \in E_\mathcal{P}$ is \textit{$\beta$-bounded} if $\hat{w}_{ij}  \leq \beta w^*_{ij}$, otherwise it is \textit{$\beta$-unbounded}. 
If all edges in $G_\mathcal{P}$ were $\beta$-bounded, finding an MST in $G_\mathcal{P}$ with respect to $\hat{w}$ would produce a $\beta$-approximation for the MST problem in $G_\mathcal{P}$ with respect to $w^*$. In general we cannot guarantee all edges will be $\beta$-bounded, as this would imply Algorithm~\ref{alg:mfcapprox} is a subquadratic $\beta$-approximation algorithm for MFC, contradicting Theorem~\ref{thm:hard}. Nevertheless, any $\beta$-bounded edge that Algorithm~\ref{alg:mfcapprox} includes in its MST of $G_\mathcal{P}$ is easy to bound in terms of the optimal edge weights $w^*$. 
If we include a $\beta$-unbounded edge in our MST of $G_\mathcal{P}$, we can no longer bound its weight in terms of an optimal solution to MFC. However, the following lemma shows that its weight can be bounded in terms of the weight of the initial forest.
\begin{lemma}
	\label{lem:unboundededges}
	Let $P_i$ and $P_j$ be an arbitrary pair of components and let $\beta > 1$. If $\hat{w}_{ij} > \beta w^*_{ij}$, then 
	\begin{align}
		\hat{w}_{ij} <
		%		\frac{\beta}{\beta - 1} \min \{\alpha_i, \alpha_j\} \leq 
		\frac{\beta}{\beta - 1}  \min\{w_\mathcal{X}(T_i), w_\mathcal{X}(T_j)\}.
	\end{align}
\end{lemma}
\begin{proof}
	For each $i \in \{1, 2, \hdots, t\}$, we denote the maximum distance between a point in $P_i$ and its component representative $s_i$ by $\alpha_i = \max_{x \in P_i} d(x, s_i)$.
	Let $x_i^* \in P_i$ and $x_j^* \in P_j$ be points satisfying $d(x_{i}^*, x_{j}^*) = d(P_i, P_j) = w_{ij}^*$.
	We use the (reverse) triangle inequality and the definition of $\alpha_i$ to see that:
	\begin{align*}
		d(x_i^*, x_j^*) = \min_{x_j \in P_j} d(x_i^*, x_j) \geq  \min_{x_j \in P_j} d(s_i, x_j) - d(s_i, x_i^*) \geq  \min_{x_j \in P_j} d(s_i, x_j) - \alpha_i = w_{j \rightarrow i} - \alpha_i \geq \hat{w}_{ij} - \alpha_i.
	\end{align*}
	Similarly we can show that 
	\begin{align*}
		d(x_i^*, x_j^*) = \min_{x_i \in P_i} d(x_i, x_j^*) \geq  \min_{x_i \in P_i} d(s_j, x_i) - d(s_j, x_j^*)  \geq  \min_{x_i \in P_i} d(s_j, x_i) - \alpha_j = w_{i \rightarrow j} - \alpha_j \geq \hat{w}_{ij} - \alpha_j.
	\end{align*}
	In other words, we have the bound $w^*_{ij} \geq \hat{w}_{ij} - \min\{\alpha_i, \alpha_j\}$. Combining this with the assumption that $\hat{w}_{ij} > \beta w_{ij}^*$ gives
	\begin{align*}
		\hat{w}_{ij} &> \beta w_{ij}^* \geq \beta\hat{w}_{ij} - \beta\min \{\alpha_i, \alpha_j\} \implies \frac{\beta}{\beta -1 } \min \{\alpha_i, \alpha_j\}> \hat{w}_{ij}.
	\end{align*}
	The proof follows from the observation that $\alpha_i \leq w_\mathcal{X}(T_i)$. To see why, note that there exists some ${x} \in P_i$ such that $d({x}, s_i) = \alpha_i$. Since $T_i$ is a spanning tree of $P_i$, it must contain a path from $s_i$ to ${x}$ with sum of edge weights at least $\alpha_i$.
\end{proof}
Our main approximation guarantees rely on Lemma~\ref{lem:unboundededges}, as well as two other simple supporting observations. This first amounts to the observation that a tree has arboricity and degeneracy 1.
\begin{observation}
	\label{lem:treeorient}
	If $T = (V,E_T)$ is an undirected tree, there is a way to orient edges in such a way that every node has at most one outgoing edge.
\end{observation}
\begin{proof}
	The proof is constructive. Define an iterative algorithm that removes a minimum degree node at each step and deletes all its incident edges. Orient the deleted edges so that they start at the node that was removed. Note that a tree always contains a node of degree 1, and removing such a node leads to another tree with one fewer node. Thus, this procedure will orient edges of the original graph in such a way that each node has at most one outgoing edge.
\end{proof}
The other supporting result deals with MSTs in a graph that includes edges of weight zero.
\begin{observation}
	\label{obs:Z}
	Let $w^{(1)} \colon E \rightarrow \mathbb{R}^+$ and $w^{(2)} \colon E \rightarrow \mathbb{R}^+$ be two nonnegative weight functions for an undirected graph $G = (V,E)$. Assume there exists an edge set $Z \subseteq E$ such that 
	\begin{equation*}
		w^{(1)}(i,j) = w^{(2)}(i,j) = 0 \text{ for every $(i,j) \in Z$}.
	\end{equation*}
	Then there exist spanning trees $M_1$ and $M_2$ for $G$ such that $M_i$ is an MST for $G$ with respect to $w^{(i)}$ for $i \in \{1,2\}$, and $M_1 \cap Z = M_2 \cap Z$.
\end{observation}
\begin{proof}
	The proof is constructive. Recall that Kruskal's algorithm finds an MST by ordering edges by weight (starting with the smallest and breaking ties arbitrarily in the ordering) and then greedily adds each edge a growing spanning tree if and only if it connects two previously disconnected components. 
	Fix an arbitrary ordering $\sigma_Z$ of edges in $Z$. When applying Kruskal's algorithm to find minimum spanning trees of $G$ with respect to $w^{(1)}$ and $w^{(2)}$, we can choose orderings for these functions that exactly coincide for the first $|Z|$ edges visited. Namely, we place edges in $Z$ first, using the order given by $\sigma_Z$. 
	% The remaining edges $E\backslash Z$ may be ordered differently in the orderings for $w^{(1)}$ and $w^{(2)}$. 
	The first $|Z|$ steps of Kruskal's algorithm will be identical when building MSTs with respect to $w^{(1)}$ and $w^{(2)}$. Thus, if $M_1$ and $M_2$ are the spanning trees obtained for $w^{(1)}$ and $w^{(2)}$ respectively using this approach, we know these trees will include the same set of edges from $Z$ and discard the same set of edges from $Z$, i.e., $M_1\cap Z = M_2 \cap Z.$
\end{proof}

\subsection{Main approximation guarantees}
Let $T^*_\mathcal{P}$ represent an MST of $G_\mathcal{P}$ with respect to the optimal weight function $w^* \colon E_\mathcal{P} \rightarrow \mathbb{R}^+$ and $\hat{T}_\mathcal{P}$ represent an MST of $G_\mathcal{P}$ with respect to the approximate weight function $\hat{w} \colon E_\mathcal{P} \rightarrow \mathbb{R}^+$. The edges of $T^*_\mathcal{P}$ map to a set of edges $M^*$ in $G_\mathcal{X}$ that optimally solves the metric MST completion problem, and the edges in $\hat{T}_\mathcal{P}$ map to an edge set $\hat{M}$. The weight of these edges is given by:
\begin{align*}
	w_\mathcal{X}(M^*) &= w^*(T^*_\mathcal{P}) \\
	w_\mathcal{X}(\hat{M}) &= \hat{w}(\hat{T}_\mathcal{P}).
\end{align*}
Let $T^*$ be the spanning tree of $G_\mathcal{X}$ defined by combining $\bigcup_{i = 1}^t T_i$ with $M^*$ and $\hat{T}$ be the spanning tree (returned by Algorithm~\ref{alg:mfcapprox}) that combines $\bigcup_{i = 1}^t T_i$ with $\hat{M}$. These have weights given by
\begin{align}
	\label{eq:Tstar}
	w_\mathcal{X}(T^*) &= w^*(T_\mathcal{P}^*) +  \sum_{i = 1}^t w_\mathcal{X}(T_i) \\
	\label{eq:That}
	w_\mathcal{X}(\hat{T}) &= \hat{w}(\hat{T}_\mathcal{P}) +  \sum_{i = 1}^t w_\mathcal{X}(T_i) .
\end{align}
We are now ready to prove the approximation guarantee for \textsf{MFC-Approx}.
\begin{theorem}
	\label{thm:main}
	The spanning tree $\hat{T}$ returned by Algorithm~\ref{alg:mfcapprox} satisfies
	\begin{equation*}
		w_\mathcal{X}(T^*) \leq w_\mathcal{X}(\hat{T}) \leq \beta w_\mathcal{X}(T^*)
	\end{equation*}
	for $\beta = (3 + \sqrt{5})/2 < 2.62$. 
\end{theorem}
\begin{proof}
	For our analysis we consider two hypothetical weight functions $w^*_0$ and $\hat{w}_0$ for $G_\mathcal{P} = (V_\mathcal{P}, E_\mathcal{P})$, defined by zeroing out the $\beta$-unbounded edges in $w^*$ and $\hat{w}$:
	\begin{align*}
		w^*_0(v_i, v_j) &= \begin{cases} 
			w^*_{ij} & \text{ if $(v_i,v_j)$ is $\beta$-bounded, i.e., $\hat{w}_{ij} \leq \beta w^*_{ij}$} \\
			0 & \text{ otherwise} 
		\end{cases}\\
		\hat{w}_0(v_i, v_j) &= \begin{cases} 
			\hat{w}_{ij} & \text{ if $(v_i,v_j)$ is $\beta$-bounded, i.e., $\hat{w}_{ij} \leq \beta w^*_{ij}$} \\
			0 & \text{ otherwise}. 
		\end{cases}
	\end{align*}
	By Observation~\ref{obs:Z}, there exist spanning trees $T_0^*$ and $\hat{T}_0$ for $G_\mathcal{P}$ that are optimal with respect to $w_0^*$ and $\hat{w}_0$, respectively, which contain the same exact set of $\beta$-unbounded edges. Let $U$ represent this set of $\beta$-unbounded edges in $\hat{T}_0$ and $T^*_0$. Let $B^*$ be the set of $\beta$-bounded edges in $T^*_0$ and $\hat{B}$ be the set of $\beta$-bounded edges in $\hat{T}_0$. Because $\hat{T}_\mathcal{P}$ is an MST with respect to $\hat{w}$ we know that:
	\begin{align}
		\label{eq:hats}
		\hat{w}(\hat{T}_\mathcal{P}) &\leq \hat{w}(\hat{T}_0) = \hat{w}(U) + \hat{w}(\hat{B}).
	\end{align}
	We will use this to upper bound the weight of $\hat{T}_\mathcal{P}$ in terms of $T^*$. First we claim that
	\begin{equation}
		\label{eq:Lhatbound}
		\hat{w}(\hat{B}) \leq \beta w^*(T_\mathcal{P}^*).
	\end{equation}
	This follows from the following sequence of inequalities:
	\begin{align*}
		\hat{w}(\hat{B}) 
		&= \hat{w}_0(\hat{B}) & \text{ since $\hat{w}$ and $\hat{w}_0$ coincide on $\beta$-bounded edges} \\
		&= \hat{w}_0(\hat{B}) + \hat{w}_0(U) & \text{ since $\hat{w}_0$ is zero on $\beta$-unbounded edges} \\
		&= \hat{w}_0(\hat{T}_0)  & \text{ since $\hat{T}_0 = \hat{B} \cup U$} \\
		& \leq \hat{w}_0(T_0^*) & \text{ since $\hat{T}_0$ is optimal for $\hat{w}_0$} \\
		& = \hat{w}_0(B^*) & \text{ since $\hat{w}_0$ is zero on $\beta$-unbounded edges} \\
		& \leq \beta w^*_0(B^*) & \text{ since $\hat{w}_0 \leq \beta w^*_0$ on $\beta$-bounded edges}\\
		%			&= w^*_0(B^*) + w^*_0(U) & \text{since $w_0^*$ is zero on $\beta$-unbounded}  \\
		&= \beta w^*_0(T_0^*) & \text{since $T_0^* = B^* \cup U$ and $w^*_0(U) = 0$}  \\
		&\leq \beta w^*_0(T_\mathcal{P}^*) & \text{since $T_0^*$ is optimal for $w_0^*$}  \\
		& \leq \beta w^*(T_\mathcal{P}^*) & \text{ since $w^*_0 \leq w^*$ for all edges.}
	\end{align*}	
	Next we bound $\hat{w}(U)$. From Lemma~\ref{lem:unboundededges}, we know that $\hat{w}_{ij} \leq {\beta}/({\beta-1}) \min \{w_\mathcal{X}(T_i), w_\mathcal{X}(T_j)\}$ for every $(v_i, v_j) \in U$. Because $\hat{T}_\mathcal{P}$ is a tree on $G_\mathcal{P}$, we know by Observation~\ref{lem:treeorient} that we can orient its edges in such a way that each node in $V_\mathcal{P} = \{v_1, v_2, \hdots, v_t\}$ has at most one outgoing edge. 
	We can therefore assign each $(v_i, v_j) \in U$ to one of its nodes in such a way that each node in $V_\mathcal{P}$ is assigned at most one edge from $U$. Assume without loss of generality that we write edges in such a way that edge $(v_i, v_j) \in U$ is assigned to node $v_i$. Thus,
	\begin{equation}
		\label{eq:small}
		\hat{w}(U) = \sum_{(v_i, v_j) \in U} \hat{w}_{ij} \leq \sum_{(v_i, v_j) \in U} \frac{\beta}{\beta-1} w_\mathcal{X}(T_i) \leq \frac{\beta}{\beta-1} \sum_{i = 1}^t w_\mathcal{X}(T_i).
	\end{equation}
	Combining these gives our final bound
	\begin{align*}
		w_\mathcal{X}(\hat{T}) 
		&= \hat{w}(\hat{T}_\mathcal{P}) + \sum_{i = 1}^t w_\mathcal{X}(T_i)  & \text{by Eq.~\eqref{eq:That}}\\
		&\leq \hat{w}(\hat{B}) + \hat{w}(U) +  \sum_{i = 1}^t w_\mathcal{X}(T_i)  & \text{by Eq.~\eqref{eq:hats}}\\
		& \leq \beta w^*(T^*_\mathcal{P}) + \left(\frac{\beta}{\beta-1} + 1\right) \sum_{i = 1}^t w_\mathcal{X}(T_i) & \text{by Eqs.~\eqref{eq:Lhatbound} and~\eqref{eq:small}}\\
		&\leq \max \left\{ \beta, \frac{\beta}{\beta - 1} + 1 \right\} \left(w^*(T_\mathcal{P}^*) +  \sum_{i = 1}^k w_\mathcal{X}(T_i) \right)  \\
		%			&\leq \max \left\{ \beta, 1 + \frac{\beta}{\beta - 1} \right\} \left(w^*(T_\mathcal{P}^*) +  \sum_{i = 1}^k w_\mathcal{X}(T_i) \right)  \\
		&= \beta w_\mathcal{X}(T^*) &\text{ by Eq.~\eqref{eq:Tstar} and our choice of $\beta$.}
	\end{align*}
	For the last step that we have specifically chosen $\beta = (3+\sqrt{5})/2$ to ensure that $\beta = 1 + \beta/(\beta-1)$, as this leads to the best approximation guarantee using the above inequalities. 
\end{proof}

Using a similar proof technique as Theorem~\ref{thm:main} we obtain the following result, showing that Algorithm~\ref{alg:mfcapprox} is a learning-augmented algorithm for metric MST whose performance depends on the $\gamma$-overlap of the initial forest. 
\begin{theorem}
	\label{thm:learning}
	Let $G_\mathcal{X}$ be an implicit metric graph and $\mathcal{P}$ be an initial partitioning with $\gamma$-overlap $\gamma = \gamma(\mathcal{P})$. Algorithm~\ref{alg:mfcapprox} returns a spanning tree of $\hat{T}$ of $G_\mathcal{X}$ that satisfies
	\begin{equation}
		w_\mathcal{X}(T_\mathcal{X}) \leq w_\mathcal{X}(\hat{T}) \leq \beta w_\mathcal{X}(T_\mathcal{X})
	\end{equation}
	where $T_\mathcal{X}$ is an MST of $G_\mathcal{X}$ and $\beta = \frac{1}{2}\left(2\gamma + 1 + \sqrt{4\gamma + 1} \right) \leq 2\gamma+ 1$.
\end{theorem}
\begin{proof}
	We use the same terminology and notation as in the proof of Theorem~\ref{thm:main}. The only difference is that we do not necessarily use $\beta = (3+ \sqrt{5})/2$. For an arbitrary $\beta \geq 1$, we can still prove in the same way that
	\begin{align}
		\label{eq:start}
		w_\mathcal{X}(\hat{T}) \leq  \beta w^*(T^*_\mathcal{P}) + \left(\frac{\beta}{\beta-1} + 1\right) \sum_{i = 1}^t w_\mathcal{X}(T_i).
	\end{align}
	The $\gamma$-overlap of the initial forest implies there exists an MST $T_\mathcal{X}$ of $G_\mathcal{X}$ satisfying:
	\begin{equation}
		\label{eq:ieratio}
		\sum_{i = 1}^t w_\mathcal{X}(T_i) = \gamma w_\mathcal{X}(I_\mathcal{X}),
	\end{equation}
	where $I_\mathcal{X}$ is the set of edges of $T_\mathcal{X}$ inside components $\mathcal{P}$ of the initial forest. Let $B_\mathcal{X}$ be the set of edges in $T_\mathcal{X}$ that cross between components, so that $w_\mathcal{X}(T_\mathcal{X}) = w_\mathcal{X}(B_\mathcal{X}) + w_\mathcal{X}(I_\mathcal{X})$. Since $T_\mathcal{X}$ is a spanning tree, $B_\mathcal{X}$ must contain a path between every pair of components, meaning that $B_\mathcal{X}$ corresponds to a spanning subgraph of the coarsened graph $G_\mathcal{P}$. Since $T_\mathcal{P}^*$ defines an MST of $G_\mathcal{P}$ with respect to $w^*$, which captures the minimum distances between pairs of components, we know
	\begin{equation}
		w^*(T_\mathcal{P}^*) \leq w_\mathcal{X}(B_\mathcal{X}).
	\end{equation}
	Putting the pieces together we see that
	\begin{equation}
		w_\mathcal{X}(\hat{T}) \leq  \beta w_\mathcal{X}(B_\mathcal{X}) + \left( 1 + \frac{\beta}{\beta-1}\right) \gamma w_\mathcal{X}(I_\mathcal{X}) \leq \max \left\{\beta, \gamma\left(1 + \frac{\beta}{\beta -1}\right)  \right\} w_\mathcal{X}(T_\mathcal{X}).
	\end{equation}
	This will hold for any choice of $\beta \geq 1$. In order to prove the smallest approximation guarantee, we choose $\beta$ satisfying:
	\begin{equation*}
		\beta = \gamma\left(1 + \frac{\beta}{\beta -1}\right).
	\end{equation*}
	The solution for this equation under constraint $\beta \geq 1$ and $\gamma \geq 1$ is 
	\begin{equation*}
		\beta = \frac{1}{2}\left(2\gamma + 1 + \sqrt{4\gamma + 1} \right) \leq 2\gamma+ 1.
	\end{equation*}
\end{proof}

\subsection{Runtime analysis and practical considerations}
Algorithm~\ref{alg:mfcapprox} finds the distance between each point in $\mathcal{X}$ and each of the $t$ component representatives, for a total of $O(nt)$ distance queries. It then finds an MST of a dense graph with ${t \choose 2}$ edges, which has runtime and space requirements of $\tilde{O}(t^2)$. Thus, the algorithm has subquadratic memory and query complexity as long as $t = o(n)$. The runtime is $\tilde{O}(nt\texttt{Q}_\mathcal{X} + t^2)$ where $\texttt{Q}_\mathcal{X}$ is the complexity for one distance query in $\mathcal{X}$, which also is subquadratic as long as $t\texttt{Q}_\mathcal{X} = o(n)$. In settings where $\texttt{Q}_\mathcal{X} = \tilde{O}(1)$, the memory, runtime, and query complexity are all subquadratic as long as $t = o(n)$.

\paragraph{Full runtime using $k$-center initialization.} The practical utility of our full MST pipeline also depends on the time it takes to find an initial forest, which depends on various design choices and trade-offs when using any strategy. For intuition we provide a rough complexity analysis for the $k$-center strategy assuming an idealized case of balanced clusters. The simple $2$-approximation for $k$-center chooses an arbitrary first cluster center, and chooses the $i$th cluster center to be the point with maximum distance from the first $i - 1$ centers~\cite{gonzalez1985clustering}. This requires $O(nt)$ distance queries. We can use the cluster centers as the component representatives for \textsf{MFC-Approx}, which allows us to compute $\hat{w}$ without any additional queries. If clusters are balanced in size, we can compute minimum spanning trees for all clusters using $O(n^2/t)$ queries and memory and a runtime of $\tilde{O}(\texttt{Q}_\mathcal{X}n^2/t)$, simply by querying all inner-cluster edges and running a standard MST algorithm. In this balanced-cluster case, combining the initial forest complexity with the complexity of our MFC algorithm, the entire pipeline for finding a spanning tree takes $\tilde{O}(\texttt{Q}_\mathcal{X}(n^2/t + nt) + t^2)$ time. This is minimized by choosing $t = \sqrt{n}$ clusters, leading to a complexity that grows as $n^{1.5}$. For unbalanced clusters, one must consider different trade-offs for cluster-balancing strategies, which could be beneficial for runtime but may affect initial cluster quality.
We could also improve the runtime at the expense of initial forest quality by not computing an exact MST for each cluster. For example, we could recursively apply our entire MFC framework to find a spanning tree of each cluster.

\paragraph{Practical improvements.} There are several ways to relax our MFC framework to make our approach faster while still satisfying strong approximation guarantees. For metrics with high query complexity, we can use approximate queries with only minor degradation in approximation guarantees. For example, for high-dimensional Euclidean distance we can apply Johnson-Lindenstrauss transformations to reduce the query complexity while approximately maintaining distances. As another relaxation, we can replace the exact nearest neighbor search subroutine in \textsf{MFC-Approx} with an approximate nearest neighbor search. If for some $\varepsilon > 0$ we find a $(1+\varepsilon)$-approximate nearest neighbor in each component for every component representative $s_i$, this will make our approximation guarantees worse by at most a factor $(1+\varepsilon)$. There are also numerous opportunities for parallelization, such as parallelizing distance queries and MST computation for components. 

We can also incorporate heuristics to improve the spanning tree quality of our algorithm with little effect on runtime. As a specific example, when approximating the distance between components $P_i$ and $P_j$ of the initial forest, we could compute $\tilde{x}_i = \argmin_{x \in P_i} d(x,s_j)$ and $\tilde{x}_j = \argmin_{x \in P_j} d(x,s_i)$ and then use the following weight for the coarsened graph:
\begin{equation*}
	\tilde{w}_{ij} = \min \{d(\tilde{x}_i, s_j),d(\tilde{x}_j, s_i), d(\tilde{x}_j,\tilde{x}_i)\}.
\end{equation*}
This differs from Algorithm~\ref{alg:mfcapprox} only in that it additionally checks the distance between $d(\tilde{x}_j,\tilde{x}_i)$ to see if this provides an even closer pair of points between $P_i$ and $P_j$. Although this does not always improve results, it can never be worse in terms of approximations. Figure~\ref{fig:wijhat} provides an example where this strategy would find the optimal distance $w_{ij}^*$, which is strictly better than $\hat{w}_{ij}$. An interesting future direction is to implement this and also explore other heuristics that could improve the practical performance of our method without affecting our theoretical guarantees.

\section{Related Work}
\label{sec:related}
The most relevant related work is discussed within context throughout the manuscript. Here we provide details about connections to (and differences from) additional work on MST algorithms. 

\paragraph{MST algorithms with predictions.}
Our work shares high-level similarities with other work on learning-augmented and query-minimizing algorithms for MSTs. Erlebach et al.~\cite{erlebach2022learning} considered a setting where a (possibly erroneous) prediction is given for each edge weight in a graph (not necessarily a metric graph) and the goal is to minimize the number of (non-erroneous) edge weight queries in order to compute an exact MST. 
Berg et al.~\cite{berg2023online} considered an online setting where edge-weight predictions are all provided a priori but true weights are only revealed in an online fashion. After revealing a true weight, an irrevocable decision must be made as to whether to include the edge in the spanning tree or not. 
Bateni et al.~\cite{bateni2024metric} recently considered MST computations in metric graphs, in settings where one has access both to a weak oracle (providing a similar type of edge-weight prediction) and a strong oracle giving true distances. The authors focus on bounding the number of strong oracle queries needed to find an exact or approximate MST. These prior works share similarities with our work in their goal to minimize certain types of queries. However, they differ from our work in that the learning-augmented information takes the form of edge weight predictions, rather than an initial forest. These works also perform $\Omega(n^2)$ queries in the worst case, which is prohibitive for large $n$.
Our work is distinctive in its focus on subquadratic algorithms for approximate MSTs.

\paragraph{Algorithms for metric MST.}
Many previous papers focus on improving algorithmic guarantees for variants of the metric MST problem, especially for the special case of Euclidean distances. For $d$-dimensional Euclidean space when $d = O(1)$, an exact MST can be computed in $O(n^{2 - 2/(\lceil d/2 \rceil + 1) + \varepsilon})$ time~\cite{agarwal1990euclidean}, and a $(1+\varepsilon)$-approximate solution can be found in time $O(n \log n + (\varepsilon^{-2} \log^2 \frac{1}{\varepsilon})n)$ time, where the big-$O$ notation hides constants of the form $O(1)^d$~\cite{arya2016fast}. For high dimensional spaces, there are also known approaches for obtaining $c$-approximate Euclidean MSTs where subquadratic runtimes depend on the value of the desired approximation factor $c > 1$~\cite{har2013euclidean,harpeled2012ann}. There are also many recent improved theoretical results for parallel and streaming variants of the metric MST problem~\cite{jayaram2024massively,chen2022new,chen2023streamingemst,wang2021fast,azarmehr2024massively,march2010fast}, most of which again focus on the Euclidean case.

\paragraph{Partitioning heuristics for MSTs.}
There are several existing divide-and-conquer heuristics for finding an approximate Euclidean MST~\cite{chen2013clustering,zhong2015fast,mishra2020efficient,jothi2018fast}. Similar to our approach, these begin by partitioning the data into smaller components, using $k$-means~\cite{zhong2015fast,jothi2018fast}, finding components of a $k$-NN graph~\cite{chen2013clustering}, or via recursive partitioning~\cite{mishra2020efficient}. An MST for the entire dataset is obtained using various heuristics for connecting components internally and then connecting disjoint components. This often involves computing an MST on some form of coarsened graph as a substep. Despite high-level similarities, these methods differ from our forest completion framework in that they focus exclusively on Euclidean space, an assumption that is essential for the partitioning schemes used and the techniques for connecting components. The other major difference is that these approaches do not attempt to provide any type of approximation guarantee, which is our main focus.

\paragraph{Alternate uses of the acronym MFC.}
To avoid possible confusion, we highlight two concepts that relate to minimum spanning trees and use the acronym MFC to denote something distinct from \textsc{Metric Forest Completion}. Liu et al.~\cite{liu2014supervised} used MFC to denote \textit{MST-based Feature Clustering}, a technique they used for feature selection. Although this involved computing a minimum spanning tree, it did not consider any notion of forest completion. Prior to this, Kor et al.~\cite{kor2011tight} studied distributed algorithms for verifying whether a tree is a minimum spanning tree of a graph. For this problem, an input graph is assumed to be partitioned among many processors, and the candidate tree consists of a collection of marked edges that are partitioned into a so-called \textit{MST fragment collection} (MFC). These fragments bear a cursory resemblance with the forest used as input to our \textsc{Metric Forest Completion} problem. However, Kor et al.~\cite{kor2011tight} do not focus on metric graphs or completing initial forests. Furthermore, their problem and approach rely on evaluating a full graph in memory, which is intractable in our setting.

\section{Experiments}
We run a large number of numerical experiments on synthetic and real-world datasets to show the practical utility of our MFC framework.
The overall utility of our approach relies both on the performance of \textsf{MFC-Approx} as well as our ability to obtain good initial forests. We therefore evaluate both aspects in practice. Our experiments are designed to address the following questions:
\begin{itemize}[leftmargin=30pt,itemsep = 0pt]
	\item[\textbf{Q1}.] How does this framework perform in terms of runtime and spanning tree cost?
	\item[\textbf{Q2}.] What $\gamma$-overlap (see Eq.~\ref{eq:gamma}) can be achieved in practice by scalable partitioning heuristics?
	\item[\textbf{Q3}.] How does practical performance compare with theoretical bounds in Theorem~\ref{thm:learning}?
	\item[\textbf{Q4}.] How do differences in the structure of the data and choice of distance metric affect performance?
\end{itemize}
As a summary of our findings, we obtain significant reductions in runtime and memory use over exact methods while obtaining very good approximation guarantees (addressing \textbf{Q1}), the $\gamma$ values tend to be very small and typically bounded by a small constant (addressing \textbf{Q2}), and the true approximation to the optimal solution is even better than our theory predicts (addressing \textbf{Q3}). Regarding \textbf{Q4}, these findings persist across a wide range of dataset types (e.g., point clouds, set data, strings) and metrics (Euclidean, Jaccard distance, Hamming distance, and edit distance) on both synthetic and real-world data. One of our key findings is that our framework performs especially well (in terms of $\gamma$-overlap and approximation ratios) on datasets where there is some inherent underlying clustering structure. This is particularly valuable given that clustering applications are a primary motivating use case for large-scale MST computations.

\subsection{Implementation details and experimental setup}
We implemented our algorithms in templated C++ code which allows for easy specialization of different distance metrics. Experiments were run on a research server with two AMD EPYC 7543 32-Core Processors and 1 TB of ram running Ubuntu 20.04.1. The code was compiled with clang version 20.0.0 with O3 using libc++ version 20.0.0. Our code is publicly available on GitHub.\footnote{\url{https://github.com/tommy1019/MetricForestCompletion}}

\paragraph{Generating initial forests.}
To generate initial forests, we first partition data points using the standard greedy 2-approximation algorithm for $k$-center clustering~\cite{gonzalez1985clustering}, as this is simple, fast, and works for arbitrary metrics. We consider results for a range of different 
 component numbers $t = k$.
After partitioning the data, we compute an exact minimum spanning tree for each component, mirroring an approach used by previous divide-and-conquer algorithms for large-scale Euclidean MST approximations~\cite{jothi2018fast,zhong2015fast,mishra2020efficient}. Computing exact MSTs for all $t$ components is often the most expensive step of our MSTC framework as it requires generating all edges inside each component. Even so, we find this leads to significant runtime improvements over applying the exact algorithm to the entire dataset, while achieving extremely good approximation ratios. 

\paragraph{Baseline and evaluation criteria.}
We compare against an exact MST computed by generating all ${n \choose 2}$ edge weights then running Kruskal's algorithm. The bottleneck in Kruskal's algorithm is sorting the $O(n^2)$ edge weights, leading to an $O(n^2 \log n)$ runtime. Although more sophisticated algorithms could achieve a runtime of $O(n^2)$ for arbitrary metrics, Kruskal's algorithm is still fast (optimal up to a log factor in terms of runtime) and has the added benefit of being simple to implement and use in practice. We note furthermore that any speed-up in our baseline algorithm for the full MST would also instantly improve the runtime of our MFC approach, since we use the same solver to find MSTs for the partitions in the initial forest. 

\begin{figure}[hbt!]
	\begin{center}
		\resizebox{\columnwidth}{!}{
			\includegraphics[page=1, width=0.4\paperwidth]{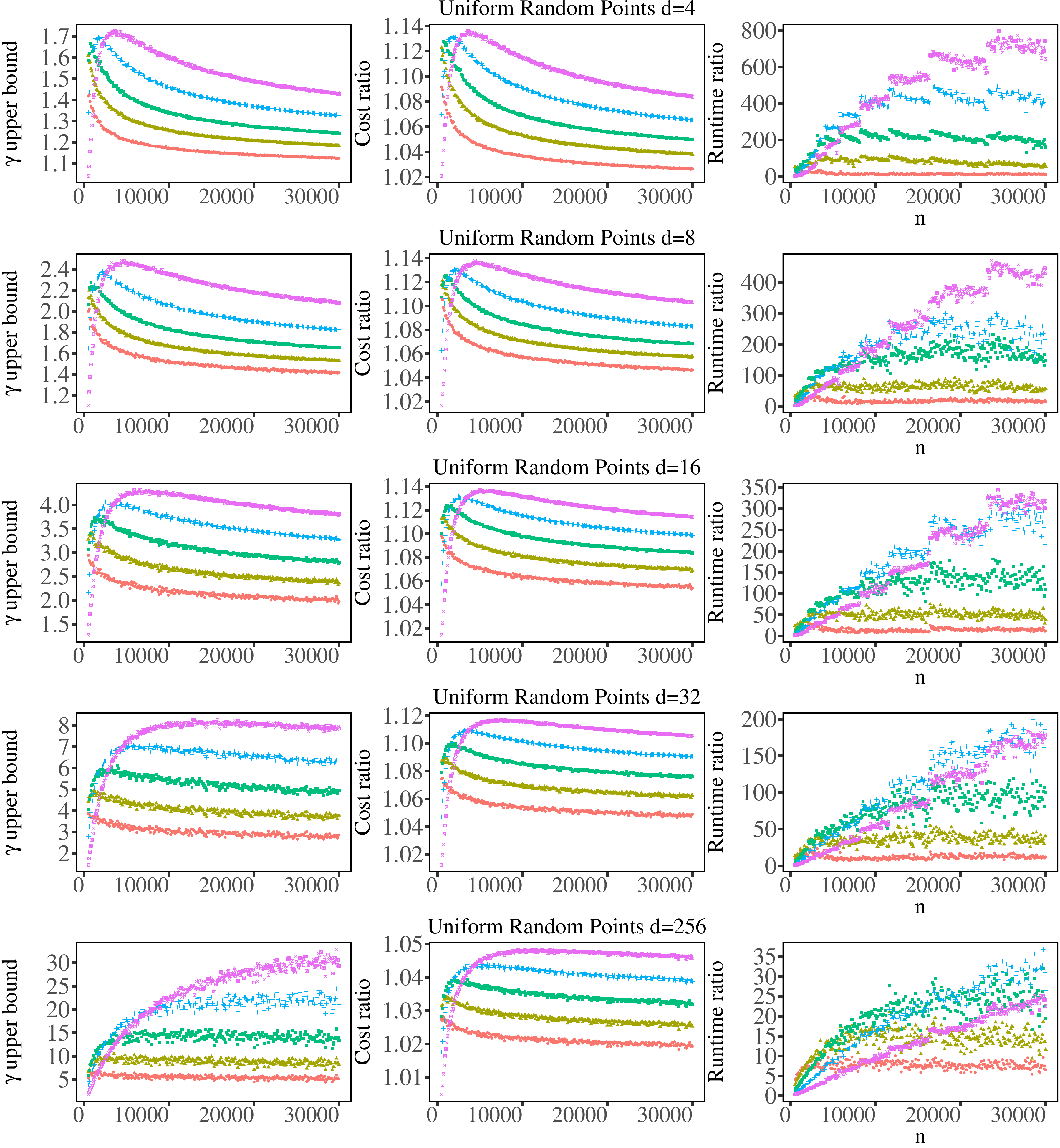}
		}
        \vspace{-12pt}
		$\begin{alignedat}{10}
			\textcolor[RGB]{229, 116, 100}{\bullet} \text{~$t=16$~} &~~~&
			\textcolor[RGB]{153, 154, 35}{\blacktriangle} \text{~$t=32$~} &~~~&
			\textcolor[RGB]{73, 178, 117}{\blacksquare} \text{~t=64~} &~~~&
			\textcolor[RGB]{67, 164, 241}{\boldsymbol{+}} \text{~$t=128$~} &~~~&
			\textcolor[RGB]{211, 108, 236}{\boldsymbol{\boxtimes}} \text{~$t=256$~} 
		\end{alignedat}$
	\end{center}
	% \vspace{-12pt}
	\caption{Results on synthetic uniform random data for dimensions $d \in \{4, 8, 16, 32, 256\}.$ Each point in each plot represents an average over $16$ sampled point clouds for a fixed $n$ and choice of component number $t$. Runtime ratio is the ratio between the runtime for the optimal MST algorithm divided by the runtime of our MFC framework (including initial forest generation). Cost ratio is the ratio between the spanning tree weight for our method and the optimal MST weight. The $\gamma$ upper bound is computed by comparing the initial forest overlap with the one optimal MST computed.}
	\label{fig:unif_grid_aux}
\end{figure}

We evaluate our MST pipeline (both the initial forest generation and the forest completion step) in terms of runtime, a bound on $\gamma$-overlap, and cost ratio (i.e., the spanning tree cost divided by the weight of an optimal MST). We find a bound $\bar{\gamma} \geq \gamma(\mathcal{P})$ on $\gamma$-overlap for our initial forest, obtained by computing the overlap between the initial forest and the MST found by the baseline.
This amounts to computing the ratio in Eq.~\eqref{eq:gamma}, except without optimizing over all possible minimum spanning trees in the denominator. When all edge weights are unique, this produces the exact value of $\gamma$. We therefore expect this to be an especially good approximation of $\gamma$ for our Euclidean point cloud datasets, since distances between random points tend to be unique (or nearly all unique).

\subsection{Experiments on Uniform Random Data}
\label{sec:nincreases}
We run a large number of experiments on uniform random data to illustrate the strong performance of our MFC framework in terms of runtime and spanning tree quality (addressing \textbf{Q1} and \textbf{Q2}) in a simple controlled setting. We consider random points clouds of size $n$ in $d$-dimensional Euclidean space where the value for a point in each dimension is drawn uniformly from $[-1,1]$. Figure~\ref{fig:unif_grid_aux} show results for $d \in \{4,8,16,32,256\}$ as $n$ increases from 1000 to 30000, using a different color for each choice of component number $t \in \{16,32,64,128,256\}$ for which we ran our MFC framework.

% \begin{figure}[t!]
% 	\begin{center}
% 		\resizebox{\columnwidth}{!}{
% 			\includegraphics[page=1, width=0.4\paperwidth]{figures/unif_all_4_to_256.pdf}
% 		}
% 		$\begin{alignedat}{10}
% 			\textcolor[RGB]{229, 116, 100}{\bullet} \text{~$t=16$~} &~~~&
% 			\textcolor[RGB]{153, 154, 35}{\blacktriangle} \text{~$t=32$~} &~~~&
% 			\textcolor[RGB]{73, 178, 117}{\blacksquare} \text{~t=64~} &~~~&
% 			\textcolor[RGB]{67, 164, 241}{\boldsymbol{+}} \text{~$t=128$~} &~~~&
% 			\textcolor[RGB]{211, 108, 236}{\boldsymbol{\boxtimes}} \text{~$t=256$~} 
% 		\end{alignedat}$
% 	\end{center}
% 	\vspace{-10pt}
% 	\caption{Results on synthetic uniform random data for dimensions $d \in \{4, 8, 16, 32, 64, 256\}.$ Each point in each plot represents an average over $16$ sampled point clouds for a fixed $n$ and choice of component number $t$. Runtime ratio is the ratio between the runtime for the optimal MST algorithm divided by the runtime of our MFC framework (including initial forest generation). Cost ratio is the ratio between the spanning tree weight for our method and the optimal MST weight. The $\gamma$ upper bound is computed by comparing the initial forest overlap with the one optimal MST computed.}
% 	\label{fig:unif_grid_aux}
% \end{figure}

As $t$ increases, we see improvements in asymptotic speedups over the exact baseline algorithm (up to a 800x speedup for $d = 4$ and 35x speedup for $d = 128$). As $t$ increases the quality of the approximate spanning tree also decreases slightly, but the cost approximation ratio still remains very good (i.e., close to 1) in all cases. The $\gamma$-overlap bound $\bar{\gamma}$ tend to be small for low dimensions, but gets large as $d$ increases (e.g., $\bar{\gamma}$ close to 30 for $d = 256$). For a given $\bar{\gamma}$, Theorem~\ref{thm:learning} guarantees roughly a $(2\bar{\gamma} + 1)$-approximation. For $d = 4$ this approximation ranges from $3$ to $5$, while for $d = 256$ it ranges from $3$ to around $60$. 
Despite this, the cost ratios obtained in practice remain far below these bounds (always very close to 1), and even improve slightly as $d$ increases. This tells us first of all that in practice our approach greatly exceeds our theoretical bounds (addressing \textbf{Q3}). We conjecture that for these high-dimensional point clouds (which lack any underlying structure), there are a large number of spanning trees that are \textit{nearly} optimal in terms of weight but are structurally very different from an optimal spanning tree (which may be unique), which could lead to high ${\gamma}$ values despite our good cost ratios in practice.

\subsection{Improved Results on Clustered Data}
\begin{figure}[t!]
	\begin{center}
		Fashion-MNIST $~~d=784$
		\includegraphics[page=1, width=0.7\paperwidth]{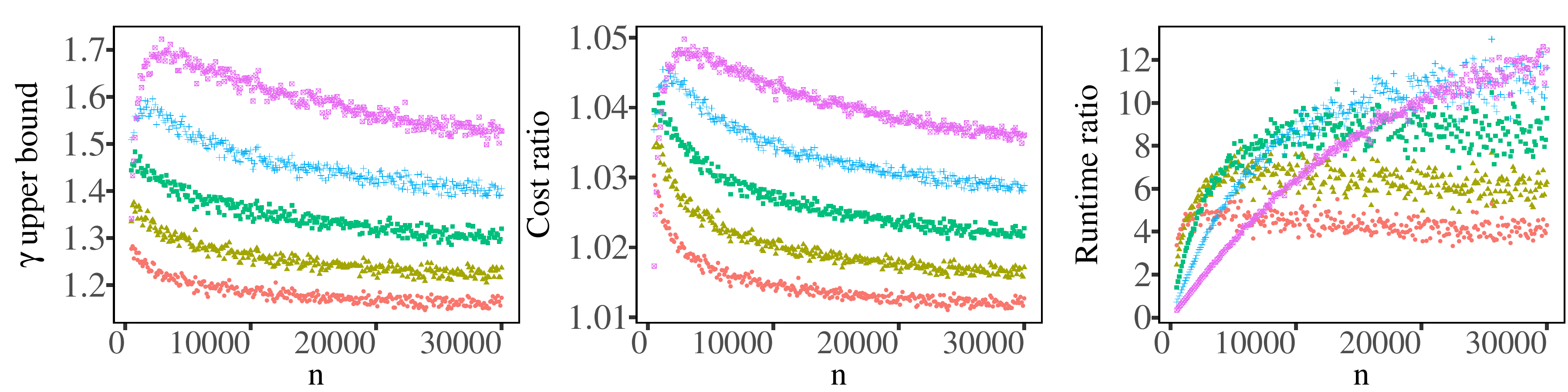}
		$\begin{alignedat}{10}
			\textcolor[RGB]{229, 116, 100}{\bullet} \text{~$t=16$~} &~~~&
			\textcolor[RGB]{153, 154, 35}{\blacktriangle} \text{~$t=32$~} &~~~&
			\textcolor[RGB]{73, 178, 117}{\blacksquare} \text{~t=64~} &~~~&
			\textcolor[RGB]{67, 164, 241}{\boldsymbol{+}} \text{~$t=128$~} &~~~&
			\textcolor[RGB]{211, 108, 236}{\boldsymbol{\boxtimes}} \text{~$t=256$~} 
		\end{alignedat}$
	\end{center}
	\vspace{-10pt}
	\caption{Results for Fashion-MNIST. Each point is the average of 16 samples for fixed $n$ and $t$.}
	\label{fig:fashion_aux}
\end{figure}
\begin{figure}[t!]
	\begin{center}
		\resizebox{\columnwidth}{!}{
			\includegraphics[page=1, width=0.45\paperwidth]{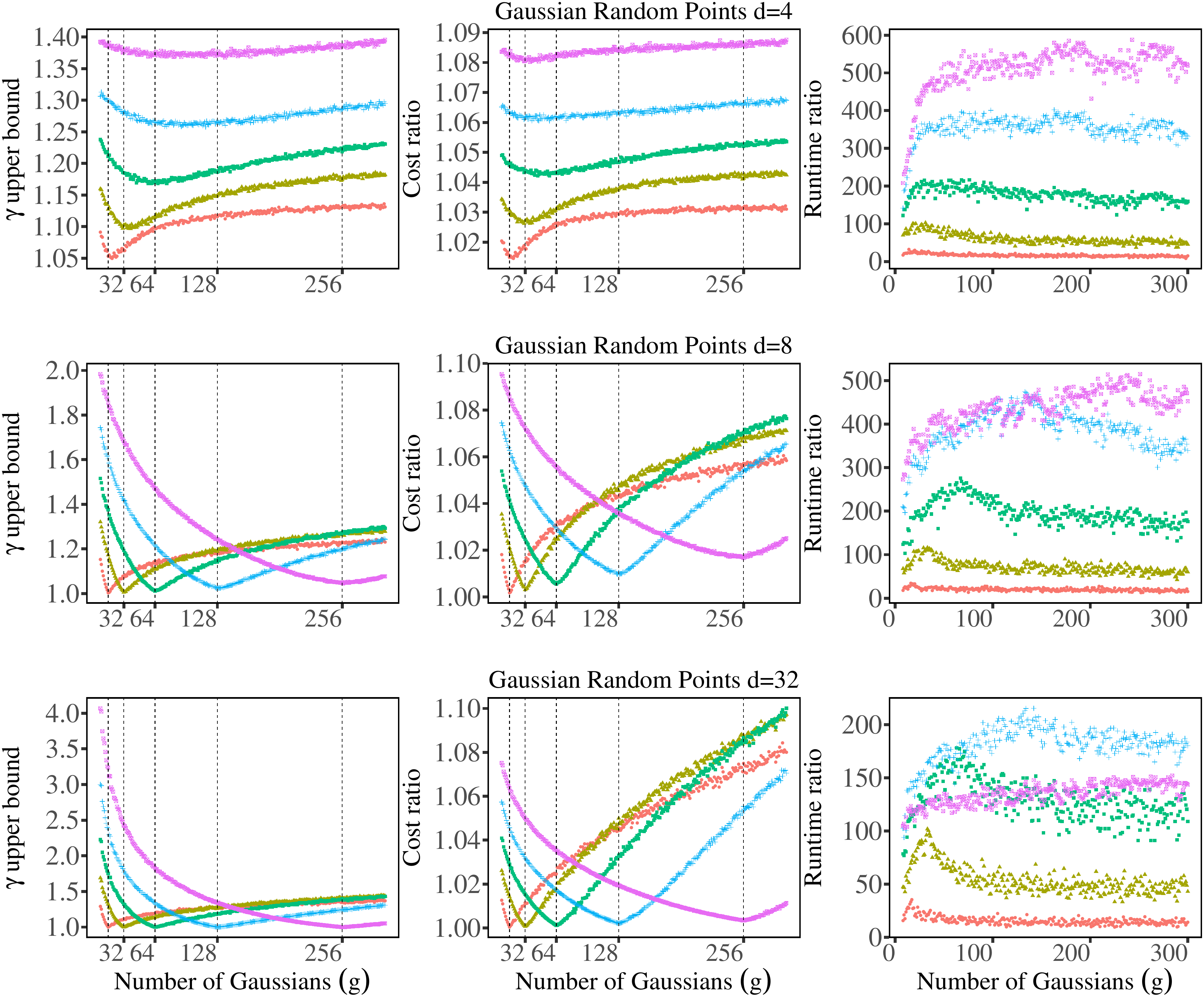}
		}
		$\begin{alignedat}{10}
			\textcolor[RGB]{229, 116, 100}{\bullet} \text{~$t=16$~} &~~~&
			\textcolor[RGB]{153, 154, 35}{\blacktriangle} \text{~$t=32$~} &~~~&
			\textcolor[RGB]{73, 178, 117}{\blacksquare} \text{~t=64~} &~~~&
			\textcolor[RGB]{67, 164, 241}{\boldsymbol{+}} \text{~$t=128$~} &~~~&
			\textcolor[RGB]{211, 108, 236}{\boldsymbol{\boxtimes}} \text{~$t=256$~} 
		\end{alignedat}$
	\end{center}
	\vspace{-10pt}
	\caption{Experimental results on synthetic data with clustering structure. Each synthetic dataset is drawn from a mixture of $g$ Gaussians in $d$-dimensional Euclidean space, with $\lfloor20000/g\rfloor$ points per Gaussian. We show results for $d \in \{4,8,32\}$  Cost ratio and $\gamma$-overlap bound $\bar{\gamma}$ are typically minimized when the number of Gaussians $g$ roughly matches the number of components $t$ used by our MSTC framework. Compared with results for uniform data (Figure~\ref{fig:unif_grid_aux}), we see our MFC framework performs even better on datasets with inherent clustering structure, especially when one has a good estimate for the number of underlying true clusters. Results shown are averages over 16 dataset samples for each $g$.}
	\label{fig:gaussian_plots}
\end{figure}
Although $\bar{\gamma}$ is large for high-dimensional uniform random data, the lack of structure in this synthetic data is atypical for most applications. For example, metric spanning trees are often computed as the first step in clustering data. Therefore, as part of our answers to Questions \textbf{Q2} (on $\gamma$ values) and \textbf{Q4} (on the structure of the data), we explore how our framework performs on datasets with some level of clustering structure. We begin by running a similar set of experiments as in Figure~\ref{fig:unif_grid_aux}, but we instead sample points uniformly at random from the Fashion-MNIST dataset~\cite{fashion_mnist}, where points represent images of clothing items. Although images do not perfectly cluster into different classes, there is still clear underlying structure (e.g., we expect images of sneakers to look different from images of trousers). Figure~\ref{fig:fashion_aux} shows that for this dataset, our MFC framework achieves slightly better cost ratios and far better $\gamma$-overlap bounds than for uniform random data (Figure~\ref{fig:unif_grid_aux}), even though the dimension of Fashion-MNIST points ($d = 784$) is much larger than the dimensions we considered for uniform data. We also observe similar trends in runtime improvements as in Figure~\ref{fig:unif_grid_aux}.

We further examine the effect of clustering structure in a controlled setting by generating mixtures of Gaussians. In more detail, we generate a mixture of $g$ Gaussians in $d$-dimensional Euclidean space (for $d \in \{4, 8,32\}$) with $g$ ranging from 8 to 300. For each Gaussian cluster we generate $\floor{20000 / g}$ points, so that $n \approx 20000$ for each dataset. The mean of each Gaussian is chosen uniformly at random from an $8$-dimensional box where each axis ranges from $-5$ to $5$. The standard deviation of each dimension of each Gaussian is chosen uniformly at random between $0.5$ and $0.8$. In practice these parameters generate good (but not perfect) clustering structure. We then run our MFC framework for each choice of initial component number $t \in \{16,32,64,128,256\}$.  Figure~\ref{fig:gaussian_plots} shows $\bar{\gamma}$ values and cost ratios for a range of component numbers $t$ as $g$ varies. Significantly, the quality of our spanning trees gets better and better as the number of components $t$ gets closer to the number of Gaussians $g$ (i.e., true number of clusters in the data). This is illustrated by bowl-shaped curves for $\bar{\gamma}$ and cost ratio that are minimized when $t \approx g$. Furthermore, even when $t$ and $g$ are not close to each other, $\bar{\gamma}$ values and cost ratios are still smaller than the ones obtained on uniform random data with the same dimension (see Figure~\ref{fig:unif_grid_aux}). This indicates that clustering structure helps our framework even when the number of underlying clusters is unknown.

\subsection{Additional Results on Real-world Data}
\begin{table}[t]
	\centering
	\caption{List of different datasets used for real world experiments}
	\label{tab:data}
	%\resizebox{\columnwidth}{!}{
	\begin{tabular}{ lll } 
		\toprule
		\textbf{Dataset} & \textbf{Type of Data} & \textbf{Distance Metric} \\
		\cmidrule(lr){1-3}
		Kosarak\cite{annBenchmarks, kosarak} & Sets & Jaccard distance \\
		MovieLens-10M\cite{annBenchmarks, movie_lens} & Sets & Jaccard distance \\
		Fashion-MNIST\cite{annBenchmarks, fashion_mnist} & $784$ dimensional points & Euclidean distance \\
		Name Dataset\cite{nameDataset2021} & Strings & Levenshtein edit distance \\
		GreenGenes\cite{greengenesDataset} & Strings & Levenshtein edit distance \\
		GreenGenes Aligned\cite{greengenesDataset} & 7682 character strings & Hamming distance\\
		\bottomrule
	\end{tabular}
	%}
\end{table}
We continue to address \textbf{Q4} by running experiments on several real-world datasets that have often been used as benchmarks for similarity search algorithms and clustering algorithms. These involve a variety of different types of data (e.g., point cloud data, sets, strings) and distance functions (e.g., Euclidean, Jaccard, Hamming, Levenshtein edit distance). For some datasets, we are unable to compute an optimal minimum spanning tree in a reasonable amount of time and space, so we subsample the data 16 times for a fixed value of $n$ and report mean results and standard deviations. We run experiments on implicit graphs with up to $n = 39774$ nodes, which means we are considering graphs with 700 million edges. Our MFC framework in fact scales up to even larger graphs, but we restrict to graphs with under 1 billion edges since in this regime we are still able to compute an optimal MST to use as a point of comparison. We summarize the dataset type and distance metric in Table~\ref{tab:data}, and provide more details for each dataset below.

\begin{itemize}[label = {},leftmargin = 10pt, itemsep = 0pt]
	\item \textbf{Cooking.}\footnote{\url{https://www.cs.cornell.edu/~arb/data/cat-edge-Cooking/}} Each data object is a set of food ingredients defining a recipe. There are 6714 ingredients and 39774 recipes. We use Jaccard distance. The original dataset comes from the \emph{What's Cooking?} Kaggle challenge~\cite{kaggle2015cooking}. We use the parsed dataset of Amburg et al.~\cite{amburg2020clustering}.
	\item \textbf{MovieLens-10M.}\footnote{\url{https://github.com/erikbern/ann-benchmarks}} This is a set of sets derived from movie ratings~\cite{movie_lens}. We consider a subset of the data provided on the ANN Benchmarks Repository~\cite{annBenchmarks}, restricting to sets with 64 items or more, in order to work with a dataset where $n \approx 30000$. We apply Jaccard distance.
	\item \textbf{Kosarak.}\footnote{\url{http://fimi.uantwerpen.be/data/kosarak.dat}} This dataset is derived from click-stream data from a Hungarian news portal~\cite{kosarak}. The original data comprises 990002 sets defined over a collection of 41270 items. We restricted to sets of size at least 40, leading to a set of 32295 sets. We apply Jaccard distance.
	\item \textbf{Fashion-MNIST.}\footnote{\url{https://github.com/zalandoresearch/fashion-mnist}} Each point in the Fashion-MNIST dataset represents a $28 \times 28$ grayscale image of an item of clothing (e.g., sneakers, trousers), encoded as a vector of size $d = 784$. The total number of class labels is 10. We use Euclidean distance for this dataset.
	\item \textbf{Names-US.}\footnote{\url{https://github.com/philipperemy/name-dataset}} Each data object is a last name for someone in the United States. This provides a natural benchmark for computing spanning trees on short sequence data using Levenshtein edit distance. We treat each UTF-8 code point as a separate character. Last names in the dataset vary in length from $0$ to $252$ characters where $0$ indicates no last name. The average length of the last names is $6.67$ with a standard deviation of $2.470$. Computing an exact minimum spanning tree for the entire dataset is infeasible, so we consider samples of names of size $n = 30000$.
	
	\item \textbf{GreenGenes-Unaligned and GreenGenes-Aligned.} GreenGenes is a chimera-checked 16S megagenomic sequence dataset~\cite{greengenesDataset}, which is frequently used as benchmarks for sequence clustering. We use version 13.5 of the data, accessed by following instructions on the USEARCH benchmarks website.\footnote{\url{https://www.drive5.com/usearch/benchmark_ggclust.html}} The data is provided in two formats: as unaligned variable length sequences (ranging in length from $1111$ to $2368$; mean length is $1401.06$ with a standard deviation of $57.083$), and as aligned sequences of a fixed length of 7682. We use Levenshtein edit distance for the unaligned sequences and Hamming distance for the aligned sequences. 
\end{itemize}

\begin{table}[t!]
	\centering
	\caption{Results for real-world datasets. Runtimes are in minutes. For the last 4 datasets we report averages and standard deviations for 16 different samples of size $n$. The last three rows for each dataset report the proportion of time spent on each step of the MFC framework.}
	\vspace{-10pt}
	\label{tab:realdata_detailed}
	\resizebox{\columnwidth}{!}{
		\begin{tabular}{ llllllll } 
			\toprule
			\textbf{Dataset} & & \textbf{OPT} & $\mathbf{t=16}$ & $\mathbf{t=32}$ & $\mathbf{t=64}$ & $\mathbf{t=128}$ & $\mathbf{t=256}$ \\
			\cmidrule(lr){1-8}
			Cooking & $ \bar{\gamma} $ & - & $ 1.770  $ & $ 2.014  $ & $ 2.222  $ & $ 2.421  $ & $ 3.144  $ \\
			& Cost Ratio & 1 & $ 1.040  $ & $ 1.051  $ & $ 1.059  $ & $ 1.069  $ & $ 1.089  $ \\
			$ n=39774$ & Runtime Ratio & 1 & $ 4.391  $ & $ 6.524  $ & $ 8.689  $ & $ 9.966  $ & $ 15.598  $ \\
			& Runtime (mins) & $24.2 \scriptstyle{ \pm 0.61 } $ & $ 5.3  $ & $ 3.5  $ & $ 2.7  $ & $ 2.3  $ & $ 1.5  $ \\
			& k-centering \% & - & $ 0.006  $ & $ 0.016  $ & $ 0.038  $ & $ 0.082  $ & $ 0.321  $ \\
			& Sub-MST \% & - & $ 0.989  $ & $ 0.971  $ & $ 0.923  $ & $ 0.826  $ & $ 0.355  $ \\
			& MFC-approx \% & - & $ 0.004  $ & $ 0.013  $ & $ 0.038  $ & $ 0.092  $ & $ 0.324  $ \\
			\cmidrule(lr){1-8}
			MovieLens & $ \bar{\gamma} $ & - & $ 1.395  $ & $ 1.513  $ & $ 1.690  $ & $ 1.955  $ & $ 2.283  $ \\
			& Cost Ratio & 1 & $ 1.010  $ & $ 1.012  $ & $ 1.017  $ & $ 1.022  $ & $ 1.028  $ \\
			$ n=33240$ & Runtime Ratio & 1 & $ 4.026  $ & $ 5.014  $ & $ 6.287  $ & $ 9.337  $ & $ 9.529  $ \\
			& Runtime (mins) & $536.0 \scriptstyle{ \pm 15.58 } $ & $ 120.3  $ & $ 96.6  $ & $ 77.0  $ & $ 51.9  $ & $ 50.8  $ \\
			& k-centering \% & - & $ 0.016  $ & $ 0.028  $ & $ 0.054  $ & $ 0.130  $ & $ 0.238  $ \\
			& Sub-MST \% & - & $ 0.978  $ & $ 0.959  $ & $ 0.912  $ & $ 0.765  $ & $ 0.508  $ \\
			& MFC-approx \% & - & $ 0.006  $ & $ 0.013  $ & $ 0.034  $ & $ 0.104  $ & $ 0.253  $ \\
			\cmidrule(lr){1-8}
			Kosarak & $ \bar{\gamma} $ & - & $ 1.289  $ & $ 1.737  $ & $ 2.281  $ & $ 2.788  $ & $ 3.129  $ \\
			& Cost Ratio & 1 & $ 1.007  $ & $ 1.013  $ & $ 1.020  $ & $ 1.025  $ & $ 1.029  $ \\
			$ n=32295$ & Runtime Ratio & 1 & $ 2.682  $ & $ 5.994  $ & $ 15.159  $ & $ 18.174  $ & $ 12.583  $ \\
			& Runtime (mins) & $182.2 \scriptstyle{ \pm 0.26 } $ & $ 68.0  $ & $ 30.4  $ & $ 12.0  $ & $ 10.0  $ & $ 14.5  $ \\
			& k-centering \% & - & $ 0.010  $ & $ 0.036  $ & $ 0.146  $ & $ 0.363  $ & $ 0.464  $ \\
			& Sub-MST \% & - & $ 0.985  $ & $ 0.935  $ & $ 0.725  $ & $ 0.339  $ & $ 0.138  $ \\
			& MFC-approx \% & - & $ 0.005  $ & $ 0.028  $ & $ 0.128  $ & $ 0.297  $ & $ 0.397  $ \\
			\cmidrule(lr){1-8}
			Names-US & $ \bar{\gamma} $ & - & $ 1.020  \scriptstyle{ \pm  0.02  } $ & $ 1.059  \scriptstyle{ \pm  0.02  } $ & $ 1.139  \scriptstyle{ \pm  0.04  } $ & $ 1.219  \scriptstyle{ \pm  0.05  } $ & $ 1.322  \scriptstyle{ \pm  0.07  } $ \\
			& Cost Ratio & 1 & $ 1.005  \scriptstyle{ \pm  0.00  } $ & $ 1.015  \scriptstyle{ \pm  0.01  } $ & $ 1.034  \scriptstyle{ \pm  0.01  } $ & $ 1.051  \scriptstyle{ \pm  0.01  } $ & $ 1.071  \scriptstyle{ \pm  0.01  } $ \\
			$ n=30000$ & Runtime Ratio & 1 & $ 0.923  \scriptstyle{ \pm  0.16  } $ & $ 0.954  \scriptstyle{ \pm  0.18  } $ & $ 0.939  \scriptstyle{ \pm  0.17  } $ & $ 1.027  \scriptstyle{ \pm  0.19  } $ & $ 1.138  \scriptstyle{ \pm  0.21  } $ \\
			& Runtime (mins) & $2.0 \scriptstyle{ \pm 0.28 } $ & $ 2.1  \scriptstyle{ \pm  0.2  } $ & $ 2.1  \scriptstyle{ \pm  0.2  } $ & $ 2.1  \scriptstyle{ \pm  0.2  } $ & $ 1.9  \scriptstyle{ \pm  0.2  } $ & $ 1.7  \scriptstyle{ \pm  0.1  } $ \\
			& k-centering \% & - & $ 0.003  \scriptstyle{ \pm  0.00  } $ & $ 0.005  \scriptstyle{ \pm  0.00  } $ & $ 0.009  \scriptstyle{ \pm  0.00  } $ & $ 0.019  \scriptstyle{ \pm  0.00  } $ & $ 0.038  \scriptstyle{ \pm  0.00  } $ \\
			& Sub-MST \% & - & $ 0.995  \scriptstyle{ \pm  0.00  } $ & $ 0.992  \scriptstyle{ \pm  0.00  } $ & $ 0.985  \scriptstyle{ \pm  0.00  } $ & $ 0.970  \scriptstyle{ \pm  0.01  } $ & $ 0.937  \scriptstyle{ \pm  0.01  } $ \\
			& MFC-approx \% & - & $ 0.001  \scriptstyle{ \pm  0.00  } $ & $ 0.003  \scriptstyle{ \pm  0.00  } $ & $ 0.005  \scriptstyle{ \pm  0.00  } $ & $ 0.011  \scriptstyle{ \pm  0.00  } $ & $ 0.024  \scriptstyle{ \pm  0.01  } $ \\
			\cmidrule(lr){1-8}
			GreenGenes-Unalign. & $ \bar{\gamma} $ & - & $ 1.388  \scriptstyle{ \pm  0.08  } $ & $ 1.460  \scriptstyle{ \pm  0.07  } $ & $ 1.460  \scriptstyle{ \pm  0.06  } $ & $ 1.402  \scriptstyle{ \pm  0.04  } $ & $ 1.314  \scriptstyle{ \pm  0.02  } $ \\
			& Cost Ratio & 1 & $ 1.092  \scriptstyle{ \pm  0.02  } $ & $ 1.104  \scriptstyle{ \pm  0.01  } $ & $ 1.095  \scriptstyle{ \pm  0.01  } $ & $ 1.075  \scriptstyle{ \pm  0.01  } $ & $ 1.055  \scriptstyle{ \pm  0.00  } $ \\
			$ n=2500$ & Runtime Ratio & 1 & $ 2.555  \scriptstyle{ \pm  0.55  } $ & $ 3.270  \scriptstyle{ \pm  0.52  } $ & $ 3.156  \scriptstyle{ \pm  0.36  } $ & $ 2.116  \scriptstyle{ \pm  0.08  } $ & $ 1.159  \scriptstyle{ \pm  0.04  } $ \\
			& Runtime (mins) & $239.3 \scriptstyle{ \pm 6.24 } $ & $ 98.0  \scriptstyle{ \pm  23.0  } $ & $ 74.9  \scriptstyle{ \pm  11.9  } $ & $ 76.7  \scriptstyle{ \pm  8.8  } $ & $ 113.2  \scriptstyle{ \pm  4.0  } $ & $ 206.5  \scriptstyle{ \pm  6.0  } $ \\
			& k-centering \% & - & $ 0.068  \scriptstyle{ \pm  0.01  } $ & $ 0.172  \scriptstyle{ \pm  0.03  } $ & $ 0.331  \scriptstyle{ \pm  0.03  } $ & $ 0.448  \scriptstyle{ \pm  0.01  } $ & $ 0.492  \scriptstyle{ \pm  0.01  } $ \\
			& Sub-MST \% & - & $ 0.885  \scriptstyle{ \pm  0.03  } $ & $ 0.692  \scriptstyle{ \pm  0.06  } $ & $ 0.377  \scriptstyle{ \pm  0.08  } $ & $ 0.134  \scriptstyle{ \pm  0.01  } $ & $ 0.047  \scriptstyle{ \pm  0.01  } $ \\
			& MFC-approx \% & - & $ 0.047  \scriptstyle{ \pm  0.01  } $ & $ 0.136  \scriptstyle{ \pm  0.03  } $ & $ 0.292  \scriptstyle{ \pm  0.05  } $ & $ 0.418  \scriptstyle{ \pm  0.01  } $ & $ 0.461  \scriptstyle{ \pm  0.01  } $ \\
			\cmidrule(lr){1-8}
			GreenGenes-Aligned & $ \bar{\gamma} $ & - & $ 1.224  \scriptstyle{ \pm  0.07  } $ & $ 1.368  \scriptstyle{ \pm  0.10  } $ & $ 1.466  \scriptstyle{ \pm  0.06  } $ & $ 1.512  \scriptstyle{ \pm  0.05  } $ & $ 1.531  \scriptstyle{ \pm  0.03  } $ \\
			& Cost Ratio & 1 & $ 1.074  \scriptstyle{ \pm  0.02  } $ & $ 1.117  \scriptstyle{ \pm  0.03  } $ & $ 1.143  \scriptstyle{ \pm  0.02  } $ & $ 1.147  \scriptstyle{ \pm  0.01  } $ & $ 1.141  \scriptstyle{ \pm  0.01  } $ \\
			$ n=30000$ & Runtime Ratio & 1 & $ 1.604  \scriptstyle{ \pm  0.51  } $ & $ 2.375  \scriptstyle{ \pm  0.82  } $ & $ 3.918  \scriptstyle{ \pm  1.99  } $ & $ 5.849  \scriptstyle{ \pm  2.83  } $ & $ 7.026  \scriptstyle{ \pm  1.69  } $ \\
			& Runtime (mins) & $33.1 \scriptstyle{ \pm 0.17 } $ & $ 22.1  \scriptstyle{ \pm  4.9  } $ & $ 15.1  \scriptstyle{ \pm  3.7  } $ & $ 9.9  \scriptstyle{ \pm  3.1  } $ & $ 6.7  \scriptstyle{ \pm  2.6  } $ & $ 5.0  \scriptstyle{ \pm  1.3  } $ \\
			& k-centering \% & - & $ 0.003  \scriptstyle{ \pm  0.00  } $ & $ 0.010  \scriptstyle{ \pm  0.00  } $ & $ 0.033  \scriptstyle{ \pm  0.02  } $ & $ 0.098  \scriptstyle{ \pm  0.05  } $ & $ 0.232  \scriptstyle{ \pm  0.05  } $ \\
			& Sub-MST \% & - & $ 0.994  \scriptstyle{ \pm  0.00  } $ & $ 0.983  \scriptstyle{ \pm  0.01  } $ & $ 0.940  \scriptstyle{ \pm  0.03  } $ & $ 0.817  \scriptstyle{ \pm  0.09  } $ & $ 0.550  \scriptstyle{ \pm  0.12  } $ \\
			& MFC-approx \% & - & $ 0.002  \scriptstyle{ \pm  0.00  } $ & $ 0.007  \scriptstyle{ \pm  0.00  } $ & $ 0.027  \scriptstyle{ \pm  0.02  } $ & $ 0.084  \scriptstyle{ \pm  0.05  } $ & $ 0.216  \scriptstyle{ \pm  0.07  } $ \\
			\cmidrule(lr){1-8}
			Fashion-MNIST & $ \bar{\gamma} $ & - & $ 1.173  \scriptstyle{ \pm  0.02  } $ & $ 1.237  \scriptstyle{ \pm  0.03  } $ & $ 1.320  \scriptstyle{ \pm  0.05  } $ & $ 1.405  \scriptstyle{ \pm  0.03  } $ & $ 1.527  \scriptstyle{ \pm  0.05  } $ \\
			& Cost Ratio & 1 & $ 1.013  \scriptstyle{ \pm  0.00  } $ & $ 1.017  \scriptstyle{ \pm  0.00  } $ & $ 1.023  \scriptstyle{ \pm  0.00  } $ & $ 1.029  \scriptstyle{ \pm  0.00  } $ & $ 1.036  \scriptstyle{ \pm  0.00  } $ \\
			$ n=30000$ & Runtime Ratio & 1 & $ 4.675  \scriptstyle{ \pm  1.42  } $ & $ 7.176  \scriptstyle{ \pm  2.55  } $ & $ 10.027  \scriptstyle{ \pm  2.85  } $ & $ 12.129  \scriptstyle{ \pm  2.64  } $ & $ 12.698  \scriptstyle{ \pm  1.83  } $ \\
			& Runtime (mins) & $11.4 \scriptstyle{ \pm 0.31 } $ & $ 2.7  \scriptstyle{ \pm  0.8  } $ & $ 1.8  \scriptstyle{ \pm  0.7  } $ & $ 1.2  \scriptstyle{ \pm  0.4  } $ & $ 1.0  \scriptstyle{ \pm  0.2  } $ & $ 0.9  \scriptstyle{ \pm  0.1  } $ \\
			& k-centering \% & - & $ 0.008  \scriptstyle{ \pm  0.00  } $ & $ 0.024  \scriptstyle{ \pm  0.01  } $ & $ 0.067  \scriptstyle{ \pm  0.02  } $ & $ 0.162  \scriptstyle{ \pm  0.03  } $ & $ 0.339  \scriptstyle{ \pm  0.05  } $ \\
			& Sub-MST \% & - & $ 0.986  \scriptstyle{ \pm  0.00  } $ & $ 0.963  \scriptstyle{ \pm  0.01  } $ & $ 0.896  \scriptstyle{ \pm  0.03  } $ & $ 0.746  \scriptstyle{ \pm  0.06  } $ & $ 0.462  \scriptstyle{ \pm  0.07  } $ \\
			& MFC-approx \% & - & $ 0.004  \scriptstyle{ \pm  0.00  } $ & $ 0.013  \scriptstyle{ \pm  0.01  } $ & $ 0.036  \scriptstyle{ \pm  0.01  } $ & $ 0.090  \scriptstyle{ \pm  0.02  } $ & $ 0.197  \scriptstyle{ \pm  0.03  } $ \\
			\bottomrule
		\end{tabular}
	}
\end{table}

Table~\ref{tab:realdata_detailed} displays cost ratios, runtimes, and $\gamma$-overlap bound $\bar{\gamma}$ for our MFC framework. The table includes details for the proportion of time spent on the three main steps of our approach: 
\begin{enumerate}[itemsep=0pt,leftmargin=10pt,label = {}]
	\item \textbf{$k$-centering}: obtaining the partitioning $\mathcal{P}$ using a greedy $2$-approximaton for $k$-centering~\cite{gonzalez1985clustering}.
	\item \textbf{Sub-MST}: Using Kruskal's algorithm to find optimal MSTs $\{T_i\}_{i=1}^t$ for each component in $\mathcal{P}$.
	\item \textbf{MFC-Approx}: Approximating the MFC problem using \textsf{MFC-Approx}.
\end{enumerate}
As anticipated, Sub-MST is typically the most expensive step. However, as $t$ increases, the $k$-centering step and the \textsf{MFC-Approx} algorithm become a significant portion of the overall runtime, and even dominate in extreme cases. We remark that the \textsf{MFC-Approx} step is always dominated by the time it takes to form the coarsened graph $G_\mathcal{P}$. The time spent finding an MST of $G_\mathcal{P}$ is negligible for all datasets and choices of $t$ we considered. 

We note several trends in the performance of our MFC framework from Table~\ref{tab:realdata_detailed}. Across all datasets (and corresponding distance metrics), the $\bar{\gamma}$ values tend to be very good: usually below 2 and never above 3.2. As expected, $\bar{\gamma}$ values get slightly worse as $t$ increases. This is also true for cost ratios, but cost ratios nevertheless remain very good in all cases (always below 1.2 and typically much better). In general, runtimes improve as $t$ increases up to a certain point. Once $t$ becomes too large, computing MSTs for the components found by $k$-centering is no longer the bottleneck. When the number of components becomes too large, the $k$-centering step and building the coarsened graph become too expensive. This is especially clear on the GreenGenes-Unaligned dataset. For this dataset we consider small subsets of size $n = 2500$ since data objects are long strings, and computing Levenshtein distances is expensive. When $t = 256$, the average partition size for the initial forest is roughly 10 points. In this case, \textsf{MFC-Approx} is extremely expensive. This algorithm is designed for situations where $t$ is asymptotically much smaller than $n$, hence Table~\ref{tab:realdata_detailed} illustrates the decrease in performance we would naturally expect when $t$ is too large.

One other exception to the general runtime trends is that for small values of $t$, running the MFC framework on the Names-US dataset is in fact slower than computing an optimal MST. For this dataset, we found that the $k$-centering step tends to form a disproportionately large component, and finding an MST of a very large component takes nearly as long as finding an MST of the entire dataset.
Combining this with the time taken on other steps in our MFC framework can lead to a larger overall runtime. This behavior (having a single large cluster) is likely tied to the fact that we are considering Levenshtein distances for short strings. This means that distances will tend to be small integer values, and many pairs of points will have the exact same distance. Our simple $k$-centering implementation can lead to imbalanced clusters in this setting where many points can be equidistant from multiple cluster center, and then one cluster is given all the points that could have just as easily gone to another cluster. Even despite this behavior, we note that runtimes do continue to increase as $t$ increases, and our MFC framework is faster for the largest $t$ values. These results could also be improved using simple strategies to avoid cluster imbalance.

\section{Conclusions and Discussion}
We have presented a new framework for finding spanning trees in arbitrary metric spaces, which is highly scalable and grounded in rigorous approximation guarantees. 
This framework is based on completing an initial forest, which can be obtained efficiently using  practical heuristics. This paper focuses on serial implementations and theoretical guarantees, as our framework already provides many advantages in this setting. At the same time, our work is strongly motivated by massive-scale clustering applications that require high-performance computing capabilities, and the algorithm we developed is highly parallelizable. A natural direction for further research is to develop parallel versions of our algorithm that can be run on a much larger scale. There are also many remaining questions in the serial setting. One direction is to try to improve on the $(\sqrt{5} + 3)/2$-approximation guarantee while still using subquadratic time, or obtain an approximation with better dependence on the $\gamma$-overlap parameter. Another direction is to prove lower bounds for the best possible approximation guarantees for subquadratic algorithms. There are also many opportunities to explore more efficient and practical methods for obtaining the initial forest that serves as the input to MFC.

%	\section*{Acknowledgments}
	
	\bibliographystyle{plain}
	\bibliography{mst-bib}

    \appendix

\end{document}